\documentclass[preprint,12pt,nopreprintline]{elsarticle}
\usepackage{amsmath,amsthm,amssymb,amsfonts}
\usepackage{graphicx}
\usepackage{subfig}
\usepackage{cuted}
\usepackage{textcomp}
\usepackage{color}
\usepackage{xcolor}
\usepackage{multirow}
\usepackage{multicol,lipsum}
\usepackage{epsfig}
\usepackage{algorithm}
\usepackage{algorithmic}
\usepackage{multirow}
\usepackage[T1]{fontenc}
\usepackage[utf8]{inputenc}
\usepackage{hyperref}
\usepackage{tabularx}
\usepackage{adjustbox}
\usepackage{mathtools}
\usepackage{bm}
\usepackage{listings}
\newtheorem{remark}{Remark}
\newtheorem{definition}{Definition}
\newtheorem{theorem}{Theorem}
\newtheorem{problem}{Problem}
\newtheorem{proposition}{Proposition}
\newtheorem{lemma}{Lemma}
\newtheorem{example}{Example}
\newtheorem{corollary}{Corollary}

\usepackage{array}
\usepackage{float}
\usepackage[short]{optidef}
\usepackage{mdframed}
\mdfdefinestyle{mystyle}{
    backgroundcolor=white!10}

\newcommand{\oomit}[1]{}

\begin{document}

\begin{frontmatter}

\title{Finite-time Safety and Reach-avoid Verification of Stochastic Discrete-time Systems}

\author{Bai Xue} 
\date{}

\affiliation{organization={Key Laboratory of System Software (Chinese Academy of Sciences) and State Key Laboratory of Computer Science, Institute of Software, Chinese Academy of Sciences},
            city={Beijing, China},
            postcode={100190}, 
            country={China}\\
            Email: xuebai@ios.ac.cn}

\begin{abstract}
This paper studies finite-time safety and reach-avoid verification for stochastic discrete-time dynamical systems. The aim is to ascertain lower and upper bounds of the probability that, within a predefined finite-time horizon, a system starting from an initial state in a safe set will either exit the safe set (safety verification) or reach a target set while remaining within the safe set until the first encounter with the target (reach-avoid verification). We introduce novel barrier-like sufficient conditions for characterizing these bounds, which either complement existing ones or fill gaps. Finally, we demonstrate the efficacy of these conditions on two simple examples.   
\end{abstract}


\end{frontmatter}

\section{Introduction}
\label{sec:intro}
Temporal verification plays a pivotal role in modern systems analysis, especially in the realm of complex systems where temporal behavior holds utmost significance \cite{prajna2007convex}. It entails a rigorous scrutiny of a system's compliance with temporal properties, including safety and reach-avoid guarantees, to ensure desired outcomes while circumventing undesirable events. Formal methods such as model checking \cite{clarke1997model,abate2010approximate} and theorem proving \cite{manna2012temporal} serve as indispensable tools in this endeavor, enabling precise and thorough analysis of temporal specifications.

Barrier certificates were initially proposed for deterministic systems as a popular formal approach to temporal verification in \cite{prajna2004safety}. They offer Lyapunov-like assurances regarding system behavior, with the mere existence of a barrier function being sufficient to establish the satisfiability of safety and/or reachability specifications  in \cite{prajna2007convex}. Subsequent efforts have focused on adapting and enhancing deterministic barrier functions, as well as broadening their applications \cite{ames2019control,taylor2020learning,xue2022reach}. However, many real-world applications are susceptible to stochastic disturbances and are thus modeled as stochastic systems. In the stochastic setting, safety verification over the infinite time horizon via barrier certificates was introduced alongside its deterministic counterpart in \cite{prajna2007framework}. Utilizing Ville’s Inequality \cite{ville1939etude}, \cite{prajna2007framework} constructed a non-negative barrier function and provided a sufficient condition for upper bounding the probability of eventually entering an unsafe region from a given initial state while remaining within a state constraint set. More recently, a new barrier function, constructed by relaxing a set of equations, was proposed for lower bounding the probability of eventually entering an unsafe or desired region from an initial state while adhering to state constraints in  \cite{xue2021reach,xue2024}. These barrier functions were further extended to determine the lower and upper bounds of the safety probability over the infinite time horizon for a specified safe set and set of initial states in \cite{yu2023safe}. Furthermore, by formulating barrier certificates via a relaxation of Bellman equations, \cite{xue2024sufficient} established necessary and sufficient conditions for lower and upper bounds the safety and reach-avoid probabilities of stochastic discrete-time systems over the infinite time horizon. In addition, under the assumption that a robust invariant set exists and the system evolves within this robust invariant set, \cite{vzikelic2023learning} proposed a new barrier certificate, termed reach-avoid supermartingale, to guarantee satisfaction of reach-avoid specifications as well as facilitate reach-avoid controllers. Barrier certificates have also been extended to infinite-time horizon probabilistic program analysis, where they verify properties such as (positive) almost-sure termination, probabilistic termination, assertion violations, and reachability (e.g., \cite{chakarov2013probabilistic,mciver2017new,moosbrugger2021automated,kenyon2021supermartingales,chatterjee2017stochastic,chatterjee2022sound,majumdar2025sound,wang2021quantitative,takisaka2021ranking}). Specially, the probabilistic termination analysis discussed in \cite{chatterjee2017stochastic,chatterjee2022sound,majumdar2025sound} exhibits interesting connections with classical analysis for stochastic discrete-time systems. When seeking a lower bound for termination probability, this analysis shares conceptual similarities with classical reach-avoid analysis. This connection arises because the termination probability essentially measures the likelihood that a program, beginning from a specific initial state, will eventually reach terminal states while preserving a stochastic invariant set throughout its execution prior to termination.
Conversely, the probabilistic termination analysis presented in \cite{majumdar2025sound}, which focuses on establishing upper bounds for termination probability, parallels classical safety analysis in stochastic discrete-time systems. In this case, the termination probability represents the chance that a program will eventually violate (escape from) a stochastic invariant set (which excludes the terminal state) when starting from a given initial state. Moreover, the application of barrier certificates has been expanded to encompass both qualitative and quantitative analysis of $\omega$-regular properties \cite{dimitrova2016probabilistic,abate2024stochastic,abate2025quantitative,henzinger2025supermartingale}.

On the other hand, finite-time temporal verification holds greater practical significance as most real-world systems operate within well-defined time constraints. Drawing inspiration from \cite{kushner1967stochastic}, the concept of a $c$-martingale was introduced in \cite{steinhardt2012finite} for stochastic continuous-time systems modelled by stochastic differential equations, enabling a controlled increase in the expected certificate value at each time step and offering an upper bound on the exit probability of leaving safe sets within bounded time horizons. Afterwards, \cite{mathiesen2022safety} proposed a computational method to find a $c$-martingale expressed by neural networks for finite-time safety verification of stochastic discrete-time systems. The $c$-martingales in \cite{mathiesen2022safety} are a typical case of the proposed barrier function in the present work. Under the assumption that the system evolves within a robust invariant set, \cite{jagtap2018temporal} extended the $c$-martingale framework to address temporal logic verification for discrete-time systems, and later, \cite{jagtap2020formal} synthesized control policies for discrete-time stochastic control systems together with a lower bound on the probability that the systems satisfy complex temporal properties. \cite{santoyo2021barrier} utilized barrier-like results introduced in \cite{kushner1967stochastic} to address the challenges of finite-time safety verification and the synthesis of safe controllers for stochastic discrete-time systems, employing semi-definite programming techniques. 
The detailed description of the verification problem in \cite{santoyo2021barrier} and its relationship with the verification problem are illustrated in Remark \ref{comparionwithsan}. On the other hand, the aforementioned works on finite-time temporal verification offer only upper bounds of the probability of reaching undesirable sets (equivalently, lower bounds of the probability of staying within desirable sets). Such works do not provide the lower bounds. This work will offer both types of bounds. Recent work by \cite{zhi2024unifying} introduced barrier functions for bounding probabilistic safety (lower and upper bounds) across infinite and finite time horizons. Its analysis, similar to \cite{jagtap2018temporal,jagtap2020formal}, relies on the assumption that system evolution remains within a robust invariant set. However, as critically examined in \cite{yu2023safe}, this strong invariance requirement constitutes a significant limitation. In contrast, the current work develops new barrier functions that eliminate the need for such a restrictive assumption. Similar to the infinite-time case, barrier certificates have also been extended to the analysis of probabilistic programs within bounded time horizons (e.g., \cite{chatterjee2017stochastic,kura2019tail,chatterjee2024quantitative}). \cite{chatterjee2017stochastic} (e.g., Lemma 3) presented a sufficient condition based on an $\epsilon$-repulsing supermartingale supported by a pure invariant, utilizing Azuma's and Hoeffding's inequalities to derive an upper bound for programs that reach a specified set exactly at a given step. Afterward, in the context where programs terminate almost surely, \cite{kura2019tail} proposed a sufficient condition for establishing lower bounds on program termination within bounded time horizons; \cite{wang2021central} developed lower and upper bounds for the tail bound problem which can also be employed to bound the probability that programs terminate within bounded time horizons. Recently, in conjunction with stochastic invariants \cite{chatterjee2017stochastic}, \cite{chatterjee2024quantitative} investigated the tail bound problem for programs that do not necessarily terminate almost surely. This approach can also be employed to establish lower bounds on the probability that programs terminate within bounded time horizons.

This paper investigates the finite-time safety and reach-avoid verification of stochastic discrete-time systems. The finite-time safety verification problem aims to compute both lower and upper bounds of the probability that a system, starting from an initial state in a safe set, will exit the safe set throughout a given bounded time interval. From a reachability perspective, it involves computing lower and upper bounds on the probability of reaching the complement of the safe set within the specified bounded time interval. Thus, it exclusively addresses safety or reachability concerns. In contrast, finite-time reach-avoid verification integrates guarantees of safety and reachability. It seeks to establish lower and upper bounds on the probability that a system, starting from an initial state in a safe set, will reach a target set within a designated bounded time interval while ensuring it remains within the safe set before reaching the target set. Although these two problems are interconnected, they differ fundamentally in essence, as we will elaborate on in the preliminaries section. We propose novel barrier-like conditions to address these two problems. These conditions either complement existing ones or fill gaps, facilitating the attainment of tight probability bounds for some systems. Finally, we demonstrate the effectiveness of the proposed conditions on two numerical examples, utilizing semi-definite programming tools.

The main contributions of this work are summarized below. 
\begin{enumerate}
    \item This work studies the finite-time safety and reach-avoid verification of stochastic discrete-time systems. Compared with existing works \cite{santoyo2021barrier,mathiesen2022safety,kushner1967stochastic}, which merely offered upper bounds on the relevant probabilities, our contribution goes beyond by providing both lower and upper bounds. These bounds deepen our understanding of the system's characteristics and yield a more accurate estimation of the true probability, thus enhancing the overall rigour and precision of the finite-time safety and reach-avoid analysis. In addition, obtaining both lower and upper bounds also facilitates evaluating their mutual tightness in practice.
    \item The proposed barrier-like conditions for upper bounding the probabilities complement existing ones in \cite{mathiesen2022safety,kushner1967stochastic}, facilitating the gain of sharper upper bounds for some systems. 
    \item Unlike the approach in \cite{zhi2024unifying}, our proposed barrier-like conditions for bounding probabilities in finite-time safety verification eliminate the need for the strong invariance assumption. Additionally, this work extends these barrier-like conditions to finite-time reach-avoid verification, providing both lower and upper probability bounds.
    \end{enumerate}

This paper is structured as follows: in Section \ref{sec:pre}, we formulate the finite-time safety and reach-avoid verification problems. In Sections \ref{sec:comp} and \ref{sec:comp1}, we introduce our barrier-like conditions for addressing these two problems, respectively. After demonstrating the performance of proposed conditions on two examples in Section \ref{sec:ex}, we conclude this paper in Section \ref{sec:con}.

\section{Preliminaries}
\label{sec:pre}
We start the exposition by a formal introduction of discrete-time systems subject to stochastic disturbances and finite-time safety/reach-avoid verification problems of interest. Before posing the problem studied, let me introduce some basic notions used throughout this paper: $\mathbb{R}$ denotes the set of real values; $\mathbb{N}$ denotes the set of nonnegative integers; $\mathbb{N}_{\leq k}$ is the set of non-negative integers being less than or equal to $k$; for sets $\Delta_1$ and $\Delta_2$, $\Delta_1\setminus \Delta_2$ denotes the difference of sets $\Delta_1$ and $\Delta_2$, which is the set of all elements in $\Delta_1$ that are not in $\Delta_2$;  $1_A(\bm{x})$ denotes the indicator function in the set $A$,
where, if $\bm{x}\in A$, then $1_A(\bm{x}) = 1$ and if $\bm{x}\notin A$, $1_A(\bm{x}) = 0$.

This paper considers stochastic discrete-time systems that are modeled by stochastic difference equations of the following form:
\begin{equation}
\label{system}
\begin{split}
\bm{x}(l+1)=\bm{f}(\bm{x}(l),\bm{\theta}(l)), \forall l\in \mathbb{N},
\end{split}
\end{equation}
where $\bm{x}(l)\in \mathbb{R}^n$ is the state at time $l$ and $\bm{\theta}(l)\in \Theta$ with $\Theta \subseteq \mathbb{R}^m$ is the stochastic disturbance at time $l$. In addition, let $\bm{\theta}(0), \bm{\theta}(1), \ldots$ be i.i.d. (independent and identically distributed) random vectors on a probability space $(\Theta,\mathcal{F},\mathbb{P})$, and take values in $\Theta$ with the following probability distribution: for any measurable set $B\subseteq \Theta$, 
\[{\rm Prob}(\bm{\theta}(l)\in B)=\mathbb{P}(B), \quad\forall l\in \mathbb{N}.\]
The expectation $\mathbb{E}[\cdot]$ is defined with respect to the probability measure $\mathbb{P}$.

Before defining trajectories of system \eqref{system}, we define a disturbance signal.
\begin{definition}
A disturbance signal $\pi$ is a sample path of the stochastic process  $\{\bm{\theta}(i)\colon \Theta\rightarrow \Theta,i\in \mathbb{N}\}$, which is defined on the canonical sample space $\Theta^{\infty}$, endowed with its product topology $\mathcal{B}(\Theta^{\infty})$, with the probability measure $\mathbb{P}^{\infty}$. The expectation associated with the probability measure $\mathbb{P}^{\infty}$ is denoted by $\mathbb{E}^{\infty}[\cdot]$.
\end{definition}

A disturbance signal $\pi$ together with an initial state $\bm{x}_0\in \mathbb{R}^n$ induces a unique discrete-time trajectory as follows.
\begin{definition}
Given a disturbance signal $\pi$ and an initial state $\bm{x}_0\in \mathbb{R}^n$, a trajectory of system \eqref{system} is denoted as  $\bm{\phi}_{\pi}^{\bm{x}_0}(\cdot)\colon\mathbb{N}\rightarrow \mathbb{R}^n$ with $\bm{\phi}_{\pi}^{\bm{x}_0}(0)=\bm{x}_0$, i.e.,
$\bm{\phi}_{\pi}^{\bm{x}_0}(l+1)=\bm{f}(\bm{\phi}_{\pi}^{\bm{x}_0}(l),\bm{\theta}(l)), \forall l\in \mathbb{N}.$
\end{definition}

In this study, we address two verification problems for the system governed by \eqref{system} with a finite time horizon $[0, N]$, where $N \in \mathbb{N}$. The first problem pertains to a finite-time safety verification problem, examining the likelihood of the system exiting the safe set $\mathcal{X}$ throughout its evolution over $[0,N]$, starting from $\bm{x}_0\subseteq \mathcal{X}$. The second one involves a finite-time reach-avoid task, focusing on the probability that the system enters the target set $\mathcal{X}_r \subset \mathcal{X}$ safely within the time horizon $[0, N]$, given an initial state $\bm{x}_0 \in \mathcal{X}\setminus \mathcal{X}_r$.

\begin{problem}[Finite-time Safety Verification]
\label{safe}
Given a finite time interval $[0,N]$ with $N\in \mathbb{N}$, a safe set $\mathcal{X}$, and an initial state $\bm{x}_0\in \mathcal{X}$, the finite-time safety verification problem is to compute  $\epsilon_1 \in [0,1]$ and $\epsilon_2\in [0,1]$ which are respectively the lower and upper bounds of the exit probability, with which system \eqref{system} starting from $\bm{x}_0$ will exit the safe set $\mathcal{X}$ within $[0,N]$, i.e., 
 \[\epsilon_1\leq \mathbb{P}^{\infty}\left(
\begin{aligned}
\exists k\in \mathbb{N}_{\leq N}. \bm{\phi}_{\pi}^{\bm{x}_0}(k)\notin \mathcal{X}
\end{aligned}
\middle|  \bm{x}_0\in \mathcal{X}
\right)\leq \epsilon_2.\]
\end{problem} 

\begin{problem}[Finite-time Reach-avoid Verification]
\label{ravoid}
Given a finite time interval $[0,N]$ with $N\in \mathbb{N}$, a safe set $\mathcal{X}$, a target set $\mathcal{X}_r\subseteq \mathcal{X}$, and an initial state $\bm{x}_0\in \mathcal{X}\setminus \mathcal{X}_r$,  the finite-time reach-avoid verification problem is to compute $\epsilon_1 \in [0,1]$ and $\epsilon_2\in [0,1]$ which are respectively the lower and upper bounds of the reach-avoid probability,  with which system \eqref{system} starting from $\bm{x}_0$ will enter the target set $\mathcal{X}_r$ within $[0,N]$ while staying inside the set $\mathcal{X}$ before the first target hitting time, i.e., 
\[\epsilon_1\leq \mathbb{P}^{\infty}\left(
\begin{aligned}
&\exists k\in \mathbb{N}_{\leq N}. \bm{\phi}_{\pi}^{\bm{x}_0}(k)\in \mathcal{X}_r\\
&\wedge \forall l\in \mathbb{N}_{\leq k}. \bm{\phi}_{\pi}^{\bm{x}_0}(l)\in \mathcal{X}
\end{aligned}
\middle|  \bm{x}_0\in \mathcal{X}\setminus \mathcal{X}_r
\right)\leq \epsilon_2.\]
\end{problem}

Like stochastic barrier functions based methods, which formulate sufficient conditions for mainly computing upper bounds of the probability of reaching unsafe sets eventually,  we in this paper will propose barrier-like conditions for addressing Problem \ref{safe} and \ref{ravoid}. 

It is noteworthy to mention that by treating system \eqref{system} within an extended domain $\widehat{\mathcal{X}}$, which includes the safe set $\mathcal{X}$ and all one-step reachable states from $\mathcal{X}$, and interpreting $\widehat{\mathcal{X}} \setminus \mathcal{X}$ as a target region as shown in the sequel, the conditions derived for Problem \ref{ravoid} can be adapted to tackle Problem \ref{safe} since $\mathbb{P}^{\infty}(\exists k\in \mathbb{N}_{\leq N}. \bm{\phi}_{\pi}^{\bm{x}_0}(k)\notin \mathcal{X} \mid  \bm{x}_0\in \mathcal{X})=\mathbb{P}^{\infty}(\exists k\in \mathbb{N}_{\leq N}. \bm{\phi}_{\pi}^{\bm{x}_0}(k)\in \widehat{\mathcal{X}}\setminus \mathcal{X} \wedge  \forall l\in \mathbb{N}_{\leq k}. \bm{\phi}_{\pi}^{\bm{x}_0}(l)\in \widehat{\mathcal{X}} \mid  \bm{x}_0\in \mathcal{X})$. However, Problem \ref{safe} and \ref{ravoid} are different. From a reachability perspective, the discrepancy arises from whether the target set is contained within the safe set $\mathcal{X}$ or not. When the target set is outside the safe set, as in the safety verification problem in Problem \ref{safe}, reaching the target set $\widehat{\mathcal{X}}\setminus \mathcal{X}$ is equivalent to exiting the safe set $\mathcal{X}$. However, this  does not hold for the reach-avoid problem in Problem \ref{ravoid}. Consequently, the conditions derived for Problem \ref{ravoid} should be refined to address Problem \ref{safe}. Thus, to underscore the improvements and maintain clarity, we treat these problems as separate entities in this paper.

\section{Finite-time Safety Verification}
\label{sec:comp}
In this section, we present our sufficient conditions for characterizing upper and lower bounds on the probability in Problem  \ref{safe}. The sufficient condition for lower bounds is formulated in Subsection \ref{sub:sclb1} and the one for upper bounds is introduced in Subsection \ref{sub:scub1}.

In accordance with the methodology in \cite{yu2023safe}, we define a switched system that is constructed by freezing the dynamics of system \eqref{system} upon exiting the safe set $\mathcal{X}$. This switched system will facilitate the construction of sufficient conditions for lower and upper bounding the exit probability in the sequel.
\begin{definition}
\label{switch}
The switched  stochastic discrete-time system, which is built upon system \eqref{system}, is a quadruple $(\widetilde{\mathcal{L}},\widetilde{\mathcal{X}},\bm{x}_0, \widetilde{\bm{f}})$ with the following components:
\begin{enumerate}
\item[-] $\widetilde{\mathcal{L}}=\{1,2\}$ is a set of two locations;
\item[-] $\widetilde{\mathcal{X}}\subseteq \mathbb{R}^n$ is the state constraint set;
\item[-] $\bm{x}_0\in \widetilde{\mathcal{X}}$ is the initial state;
\item[-] $\widetilde{\bm{f}}(\cdot,\cdot)\colon\mathbb{R}^n \times \Theta \rightarrow \mathbb{R}^n$, where \[\widetilde{\bm{f}}(\bm{x},\bm{\theta})=\sum_{i=1}^2 1_{\widetilde{\mathcal{X}}_i}(\bm{x})\widetilde{\bm{f}}_i(\bm{x},\bm{\theta})\] with $\widetilde{\bm{f}}_1(\bm{x},\bm{\theta})=\bm{f}(\bm{x},\bm{\theta})$ and  $\widetilde{\bm{f}}_2(\bm{x},\bm{\theta})=\bm{x}$, and $1_{\widetilde{\mathcal{X}}_i}(\bm{x})$ is the indicator function of the set $\widetilde{\mathcal{X}}_i$, i.e., $1_{\widetilde{\mathcal{X}}_i}(\bm{x})=1$ if $\bm{x}\in \widetilde{\mathcal{X}}_i$; otherwise, $1_{\widetilde{\mathcal{X}}_i}(\bm{x})=0$,
\end{enumerate}
where 
\begin{enumerate}
\item $\widetilde{\mathcal{X}}$ is a set satisfying $\widetilde{\Omega}\subset\widetilde{\mathcal{X}}$, where $\widetilde{\Omega}$ is the union of $\mathcal{X}$ and
all reachable states starting from $\mathcal{X}$ in one step, i.e.,  
\[\widetilde{\Omega}=\{\bm{x}'\in \mathbb{R}^n\mid \bm{x}'=\bm{f}(\bm{x},\bm{\theta}), \bm{x}\in \mathcal{X}, \bm{\theta}\in \Theta\}\cup \mathcal{X};\]
\item $\widetilde{\mathcal{X}}_1=\mathcal{X}$;
\item $\widetilde{\mathcal{X}}_2=\widetilde{\mathcal{X}}\setminus \mathcal{X}$.
\end{enumerate}
The evolution of the state of this system is governed by the iterative map
\begin{equation}
\label{sdss1}
\begin{split}
&\bm{x}(l+1)=\widetilde{\bm{f}}(\bm{x}(l),\bm{\theta}(l)), \forall l\in \mathbb{N},\\
&\bm{x}(0)=\bm{x}_0\in \widetilde{\mathcal{X}}.
\end{split}
\end{equation}
\end{definition}

The trajectory of system \eqref{sdss1}, induced by initial state $\bm{x}_0\in \widetilde{\mathcal{X}}$ and disturbance policy $\pi$, is denoted by 
$\widetilde{\bm{\phi}}_{\pi}^{\bm{x}_0}(\cdot)\colon \mathbb{N} \rightarrow \mathbb{R}^n$. 

It is observed that if system \eqref{system} starting from $\bm{x}_0\in \mathcal{X}$ enters the set $\widetilde{\mathcal{X}}\setminus \mathcal{X}$ initially at time $i\in \mathbb{N}$, system \eqref{sdss1} starting from $\bm{x}_0\in \mathcal{X}$ will also enter it for the first time at this time instant, and vice versa. Furthermore, when system \eqref{sdss1} transitions into the set $\widetilde{\mathcal{X}}\setminus \mathcal{X}$, it remains there indefinitely. Thus, the set $\widetilde{\mathcal{X}}$ is a robust invariant for system \eqref{sdss1} \cite{yu2023safe}, that is, trajectories of system \eqref{sdss1} originating from the set $\widetilde{\mathcal{X}}$ will never leave it. Consequently, the probability of system \eqref{sdss1} entering $\widetilde{\mathcal{X}}\setminus \mathcal{X}$ at time $t=i$ is equal to the probability of system \eqref{system} entering this set up to time $t=i$. 
\begin{lemma}
\label{lema1}
For $i\in \mathbb{N}$, $\mathbb{P}^{\infty}(\exists k\in \mathbb{N}_{\leq i}. \bm{\phi}_{\pi}^{\bm{x}_0}(k)\in \widetilde{\mathcal{X}}\setminus \mathcal{X} \mid \bm{x}_0\in \mathcal{X})=\mathbb{P}^{\infty}(\widetilde{\bm{\phi}}_{\pi}^{\bm{x}_0}(i)\in \widetilde{\mathcal{X}}\setminus \mathcal{X}  \mid \bm{x}_0\in \mathcal{X})$. Thus, $\mathbb{P}^{\infty}(\widetilde{\bm{\phi}}_{\pi}^{\bm{x}_0}(i)\in \widetilde{\mathcal{X}}\setminus \mathcal{X}  \mid \bm{x}_0\in \mathcal{X})\leq \mathbb{P}^{\infty}(\widetilde{\bm{\phi}}_{\pi}^{\bm{x}_0}(j)\in \widetilde{\mathcal{X}}\setminus \mathcal{X}  \mid \bm{x}_0\in \mathcal{X})$ if $i\leq j$.
\end{lemma}

\subsection{Sufficient Conditions for Upper Bounds}
\label{sub:sclb1}

In this subsection, we present our sufficient condition for upper bounding the probability of exiting the safe set $\mathcal{X}$ over the time horizon $[0,N]$ in Problem \ref{safe}.

The sufficient condition requires the existence of a non-negative function $v(\bm{x}): \widetilde{\mathcal{X}}\rightarrow \mathbb{R}$ defined on the set $\widetilde{\mathcal{X}}$ that satisfies two key properties: (1) it must be greater than or equal to one on the subset $\widetilde{\mathcal{X}}\setminus \mathcal{X}$, and (2) its $\alpha$-scaled expected value at the next step along the system dynamics described in \eqref{system} should not increase by more than $\alpha \beta$ relative to its current value for $\bm{x}\in \mathcal{X}$, where $\alpha\in (0,1]$ and $\beta\in (-\infty,1]$. The upper bounds derivation proceeds as follows: Lemma \ref{lema1} establishes that the probability of system \eqref{sdss1} entering $\widetilde{\mathcal{X}}\setminus \mathcal{X}$ at time $t=N$ is equal to the probability of system \eqref{system} entering this set up to time $t=N$. Then, we derive upper bounds using system \eqref{sdss1} by analyzing $\alpha$ and $\beta$. Relying on the fact that the system in \eqref{sdss1} remains stationary when starting from $\widetilde{\mathcal{X}}\setminus \mathcal{X}$, which implies $\mathbb{E}^{\infty}[v(\widetilde{\bm{\phi}}_{\pi}^{\bm{x}}(1))]=v(\bm{x})$ for $\bm{x}\in \widetilde{\mathcal{X}}\setminus \mathcal{X}$, we reformulate the constraint over the function $v(\bm{x}): \widetilde{\mathcal{X}}\rightarrow \mathbb{R}$ using system \eqref{sdss1}, requiring that the $\alpha$-scaled expected value of $v(\bm{x})$ at the next step along the system \eqref{sdss1} grows by at most $\alpha \beta$ relative to its current value $v(\bm{x})$ over the whole set $\widetilde{\mathcal{X}}$ when $\alpha\in (0,1]$ and $\beta \alpha-(\alpha-1) \in [0,1]$.  Upper bounds follow from recursive application of this reformulated constraint; when $\alpha\in (0,1]$ and $\beta \alpha-(\alpha-1) \in (-\infty,0)$, the constraint is modified to require that the expected value of $v(\bm{x})$ at the next step along the system \eqref{sdss1} grows by at most $\beta$, relative to the $(1 - \beta)$-weighted current value $v(\bm{x})$. The upper bound again emerges through recursive application of this reformulated constraint.

\begin{theorem}
    \label{pro_e2}
    If there exist a function $v(\bm{x})\colon\widetilde{\mathcal{X}}\rightarrow \mathbb{R}$, and $\alpha \in (0,1]$ and $\beta \in (-\infty,1]$ such that 
    \begin{equation}
   \label{ebr_2}
   \begin{cases}
   \mathbb{E}^{\infty}[v(\bm{\phi}_{\pi}^{\bm{x}}(1))]\leq \frac{v(\bm{x})}{\alpha} +\beta, & \forall \bm{x}\in \mathcal{X},\\
   v(\bm{x})\geq 1, & \forall \bm{x}\in  \widetilde{\mathcal{X}}\setminus \mathcal{X},\\
   v(\bm{x})\geq 0, &\forall \bm{x}\in \mathcal{X},
   \end{cases}
   \end{equation}
    then $\mathbb{P}^{\infty}\big(\exists k\in \mathbb{N}_{\leq N}. \bm{\phi}^{\bm{x}_0}_{\pi}(k) \in \widetilde{\mathcal{X}}\setminus \mathcal{X} \mid  \bm{x}_0\in \mathcal{X} \big)=\mathbb{P}^{\infty}\big( \widetilde{\bm{\phi}}^{\bm{x}_0}_{\pi}(N) \in \widetilde{\mathcal{X}}\setminus \mathcal{X} \mid  \bm{x}_0\in \mathcal{X} \big)\leq$
   \[
    \begin{cases}
       v(\bm{x}_0)+\beta N, &\text{if~}\alpha=1 \wedge \gamma \in [0,1],\\
        v(\bm{x}_0) \alpha^{-N}+\frac{(1-\alpha^{-N})\alpha \beta}{\alpha-1}, & \text{if~}\alpha\in (0,1)\wedge \gamma\in [0,1],\\
       1-(1- v(\bm{x}_0))(1-\beta)^{N},                &\text{if~}\alpha\in (0,1] \wedge \gamma \in (-\infty, 0),
 \end{cases}
 \] 
 where $\gamma=\beta \alpha -(\alpha-1)$.  
    \end{theorem}
    \begin{proof}
    1) We first prove the case of $\alpha\in (0,1)\wedge \gamma\in [0,1]$.

Since $\mathbb{E}^{\infty}[v(\widetilde{\bm{\phi}}_{\pi}^{\bm{x}}(1))]=v(\bm{x})$ for $\bm{x}\in \widetilde{\mathcal{X}}\setminus \mathcal{X}$, we can obtain 
\[\mathbb{E}^{\infty}[v(\widetilde{\bm{\phi}}_{\pi}^{\bm{x}}(1))]-\frac{1}{\alpha}v(\bm{x})=(1-\frac{1}{\alpha}) v(\bm{x}), \forall \bm{x}\in \widetilde{\mathcal{X}}\setminus \mathcal{X}.\]
Therefore, \[\mathbb{E}^{\infty}[v(\widetilde{\bm{\phi}}_{\pi}^{\bm{x}}(1))]-\frac{1}{\alpha}v(\bm{x})\leq \beta, \forall \bm{x}\in \widetilde{\mathcal{X}}\setminus \mathcal{X},\] which can be justified based on the facts that $v(\bm{x})\geq 1$ for $\bm{x}\in \widetilde{\mathcal{X}}\setminus \mathcal{X}$, $\alpha\in (0,1)$ and $\gamma \in [0,1]$. Therefore, if $v(\bm{x})$ satisfies \eqref{ebr_2}, it will satisfy
   \begin{equation}
   \label{ebs_2}
   \begin{cases}
   \mathbb{E}^{\infty}[v(\widetilde{\bm{\phi}}_{\pi}^{\bm{x}}(1))]\leq \frac{v(\bm{x})}{\alpha} +\beta, & \forall \bm{x}\in \widetilde{\mathcal{X}},\\
   v(\bm{x})\geq 1_{\widetilde{\mathcal{X}}\setminus \mathcal{X}}(\bm{x}), & \forall \bm{x}\in \widetilde{\mathcal{X}}.
   \end{cases}
   \end{equation}  
    According to \eqref{ebs_2}, we have  
    \[
    \begin{split}
    &v(\bm{x}_0)\geq 1_{\widetilde{\mathcal{X}}\setminus \mathcal{X}}(\bm{x}_0),\\
    &\alpha^{-1}v(\bm{x}_0)+\beta\geq \mathbb{E}^{\infty}[v(\widetilde{\bm{\phi}}_{\pi}^{\bm{x}_0}(1))]\geq \mathbb{E}^{\infty}[1_{\widetilde{\mathcal{X}}\setminus \mathcal{X}}(\widetilde{\bm{\phi}}_{\pi}^{\bm{x}_0}(1))],\\  
    &\ldots,\\
    &\alpha^{-N}v(\bm{x}_0)+\beta \sum_{i=0}^{N-1} \alpha^{-i} \geq  \mathbb{E}^{\infty}[v(\widetilde{\bm{\phi}}_{\pi}^{\bm{x}_0}(N))]\geq \mathbb{E}^{\infty}[1_{\widetilde{\mathcal{X}}\setminus \mathcal{X}}(\widetilde{\bm{\phi}}_{\pi}^{\bm{x}_0}(N))].
    \end{split}
    \]
    Therefore, 
 $\mathbb{P}^{\infty}\big(\widetilde{\bm{\phi}}^{\bm{x}_0}_{\pi}(N) \in \widetilde{\mathcal{X}}\setminus \mathcal{X} \mid  \bm{x}_0\in \mathcal{X} \big)=\mathbb{E}^{\infty}[1_{\widetilde{\mathcal{X}}\setminus \mathcal{X}}(\widetilde{\bm{\phi}}_{\pi}^{\bm{x}_0}(N))]\leq \alpha^{-N}v(\bm{x}_0)+\alpha \beta \frac{(1-\alpha^{-N})}{\alpha-1}$. According to Lemma \ref{lema1}, we have the conclusion.

2)  The conclusion for the case of $\alpha=1\wedge \gamma\in [0,1]$ can be justified via following the proof of the above one.

3) We will show the case of $\alpha\in (0,1]\wedge \gamma\in (-\infty,0)$.

Since $\beta <1-\frac{1}{\alpha}$, thus $\frac{1}{\alpha}<1-\beta$. Consequently, $1-\beta>0$.  Thus, we conclude that if $v(\bm{x})$ satisfies \eqref{ebr_2}, it will satisfy
  \begin{equation*}
   \begin{cases}
   \mathbb{E}^{\infty}[v(\bm{\phi}_{\pi}^{\bm{x}}(1))]\leq v(\bm{x})(1-\beta) +\beta,& \forall \bm{x}\in \mathcal{X},\\
   v(\bm{x})\geq 1, & \forall \bm{x}\in  \widetilde{\mathcal{X}}\setminus \mathcal{X},\\
   v(\bm{x})\geq 0, &\forall \bm{x}\in \mathcal{X},
   \end{cases}
   \end{equation*}
and thus it satisfies
\begin{equation*}
   \begin{cases}
   \mathbb{E}^{\infty}[v(\widetilde{\bm{\phi}}_{\pi}^{\bm{x}}(1))]\leq (1-\beta)v(\bm{x}) +\beta, & \forall \bm{x}\in \widetilde{\mathcal{X}},\\
   v(\bm{x})\geq 1_{\widetilde{\mathcal{X}}\setminus \mathcal{X}}(\bm{x}), & \forall \bm{x}\in \widetilde{\mathcal{X}}.
   \end{cases}
   \end{equation*}
Therefore, following the proof for the case $\alpha\in (0,1]\wedge \gamma\in [0,1]$, we will have the conclusion.

The proof is completed.
   \end{proof}

   In Theorem \ref{pro_e2}, the parameter $\beta$ is practically constrained to be less than or equal to one. This restriction is necessary to ensure the practical utility of the upper bounds $v(\bm{x}_0)+\beta N$ and $v(\bm{x}_0) \alpha^{-N}+\frac{(1-\alpha^{-N})\alpha \beta}{\alpha-1}$. Assigning a value greater than one to $\beta$ would lead to these two bounds exceeding unity when $N \geq 1$, rendering them ineffective.

If $\alpha=1$ and $\beta=0$, Theorem \ref{pro_e2} is equivalent to Proposition 3 in \cite{yu2023safe}, which formulates a sufficient condition for determining an upper bound of the probability of exiting the safe set $\mathcal{X}$ eventually (i.e., $\mathbb{P}^{\infty}\big(\exists k\in \mathbb{N}. \bm{\phi}^{\bm{x}_0}_{\pi}(k) \in \widetilde{\mathcal{X}}\setminus \mathcal{X} \mid  \bm{x}_0\in \mathcal{X} \big)$). The upper bound is $v(\bm{x}_0)$, which is also an upper bound of the exit probability $\mathbb{P}^{\infty}\big(\exists k\in \mathbb{N}_{\leq N}. \bm{\phi}^{\bm{x}_0}_{\pi}(k) \in \widetilde{\mathcal{X}}\setminus \mathcal{X} \mid  \bm{x}_0\in \mathcal{X} \big)$. 
On the other hand, when $\alpha\in (0,1)$ and $\gamma \in (-\infty,0)$, if $v(\bm{x}_0)<1$,  $\lim_{N\rightarrow +\infty}1-(1-v(\bm{x}_0))(1-\beta)^{N}=-\infty$ holds, implying  $\mathbb{P}^{\infty}\big(\exists k\in \mathbb{N}. \bm{\phi}^{\bm{x}_0}_{\pi}(k) \in \widetilde{\mathcal{X}}\setminus \mathcal{X} \mid  \bm{x}_0\in \mathcal{X} \big)=0$ and $\mathbb{P}^{\infty}\big(\exists k\in \mathbb{N}_{\leq N}. \bm{\phi}^{\bm{x}_0}_{\pi}(k) \in \widetilde{\mathcal{X}}\setminus \mathcal{X} \mid  \bm{x}_0\in \mathcal{X} \big)=0$. 

Theorem \ref{pro_e2} extends the conclusion drawn in \cite{kushner1967stochastic}, as it allows for a value of $\alpha$ within the interval $(0,1)$, in contrast to the requirement of $\alpha \geq  1$ in the cited work. For convenience of reference, we present the related result in \cite{kushner1967stochastic} here, which corresponds to Theorem 3 in Chapter 3 in \cite{kushner1967stochastic}.
\begin{proposition}
\label{theorem3}
    If there exists a continuous nonnegative function $v(\cdot): \widetilde{\mathcal{X}}\rightarrow \mathbb{R}$ satisfying
   \begin{equation}
   \label{ebr_33}
   \begin{cases}
   v(\bm{x}_0)<1,& \\
   \mathbb{E}^{\infty}[v(\bm{\phi}_{\pi}^{\bm{x}}(1))]\leq \frac{v(\bm{x})}{\alpha} +\beta, & \forall \bm{x}\in \mathcal{X},\\
   v(\bm{x})\geq 1, & \forall \bm{x}\in  \widetilde{\mathcal{X}}\setminus \mathcal{X},\\
   v(\bm{x})\geq 0, &\forall \bm{x}\in \mathcal{X},
   \end{cases}
   \end{equation}
where $\alpha\geq 1$ and $\beta \geq 0$, then $\mathbb{P}^{\infty}\big(\forall k\in \mathbb{N}_{\leq N}. \bm{\phi}^{\bm{x}_0}_{\pi}(k) \in \widetilde{\mathcal{X}}\setminus \mathcal{X} \mid  \bm{x}_0\in \mathcal{X} \big)\leq$
   \[
    \begin{cases}
        1-(1-v(\bm{x}_0))(1-\beta)^N, &\text{if~} \alpha>1\wedge \gamma \in (-\infty,0], \\
        v(\bm{x}_0) \alpha^{-N}+\frac{(1-\alpha^{-N})\alpha \beta}{\alpha-1},  & \text{if~}\alpha>1\wedge \gamma \in (0,\infty),\\
        v(\bm{x}_0)+\beta N, & \text{if~}\alpha=1 \wedge \gamma \in (0,\infty),
 \end{cases}
 \]
 where $\gamma=\alpha \beta-(\alpha-1)$.
 \end{proposition}

Since $\bm{x}_0\in \{\bm{x}\in \widetilde{\mathcal{X}}\mid v(\bm{x})<1\}$ and $\widetilde{\mathcal{X}}\setminus \mathcal{X}\subseteq \{\bm{x}\in \widetilde{\mathcal{X}}\mid v(\bm{x})\geq 1\}$, according to Theorem 3 in Chapter 3 in \cite{kushner1967stochastic}, the above statement holds. Similarly, the parameter $\beta$ in \eqref{ebr_33} should be less than or equal to one. When $\gamma\leq 0$ and $\alpha>1$, $\beta\leq 1-\frac{1}{\alpha}$ holds and thus $\beta\leq 1$ holds. For the case with $\gamma>0$, when $\beta>1$ and $N\geq 1$,  $v(\bm{x}_0) \alpha^{-N}+\frac{(1-\alpha^{-N})\alpha \beta}{\alpha-1}>1$ and $v(\bm{x}_0)+\beta N >1$ hold, which are ineffective upper bounds. 

When comparing constraints \eqref{ebr_2} and \eqref{ebr_33}, it becomes apparent that if there exists a function $v(\bm{x})$ along with $\alpha>1$ and $\beta\in [0,1]$ satisfying \eqref{ebr_33}, the same function $v(\bm{x})$ with $\frac{1}{\alpha}$ and $\beta\in [0,1]$ will satisfy \eqref{ebr_2}. In addition, it is observed that the barrier function in \cite{mathiesen2022safety}, referred to as c-martingales, is just a function that satisfies constraint \eqref{ebr_2} with $\alpha=1$ and $\beta\geq 0$.

\subsection{Sufficient Conditions for Lower Bounds}
\label{sub:scub1}

In this subsection, we present our sufficient condition for lower bounding the probability of exiting the safe set over the time horizon $[0,N]$ in Problem \ref{safe}, which is inspired by \cite{xue2021reach,xue2024}.

The sufficient condition requires the existence of a function $v(\bm{x})$ defined on the set $\widetilde{\mathcal{X}}$ that satisfies three key properties: (1) it admits a finite upper bound over $\widetilde{\mathcal{X}}$, (2) it is less than or equal to one on the subset $\widetilde{\mathcal{X}}\setminus \mathcal{X}$, and (3) its expected value at the next step along the system dynamics described in \eqref{system} increases over its $\alpha$-scaled current value by more than $\beta$ for $\bm{x}\in \mathcal{X}$, where $\alpha \in [1,\infty)$, and $\beta\in (1-\alpha,\infty)$. The lower bounds derivation proceeds as follows: Lemma \ref{lema1} establishes that the probability of system \eqref{sdss1} entering $\widetilde{\mathcal{X}}\setminus \mathcal{X}$ at time $t=N$ is equal to the probability of system \eqref{system} entering this set up to time $t=N$. Then, we derive the lower bounds using system \eqref{sdss1} by analyzing $\alpha$. Relying on the fact that the system in \eqref{sdss1} remains stationary when starting from $\widetilde{\mathcal{X}}\setminus \mathcal{X}$, which implies $\mathbb{E}^{\infty}[v(\widetilde{\bm{\phi}}_{\pi}^{\bm{x}}(1))]=v(\bm{x})$ for $\bm{x}\in \widetilde{\mathcal{X}}\setminus \mathcal{X}$, and noting that $v(\bm{x})$ is less than or equal to one on this subset, we reformulate the constraint on $v(\bm{x})$. Specifically, using system \eqref{sdss1}, we require that the expected value of $v(\bm{x})$ at the next step grows at least by $\beta-(\alpha-1+\beta)1_{\widetilde{\mathcal{X}}\setminus \mathcal{X}}(\bm{x})$  over its $\alpha$-scaled current value $v(\bm{x})$, for all $\bm{x}\in \widetilde{\mathcal{X}}$. For $\alpha>1$ and $\alpha=1$, the lower bounds are subsequently obtained through recursive application of the reformulated constraint, in conjunction with the upper bound on the function $v(\bm{x}): \widetilde{\mathcal{X}}\rightarrow \mathbb{R}$.

\begin{theorem}
    \label{pro_e7}
      If there exist a function $v(\bm{x})\colon\widetilde{\mathcal{X}}\rightarrow \mathbb{R}$ with $\sup_{\bm{x}\in \widetilde{\mathcal{X}}}v(\bm{x})\leq M$, $\alpha \in [1,\infty)$, and $\beta\in (1-\alpha,\infty)$ such that 
      \begin{equation}
   \label{ebr_e7}
   \begin{cases}
   \beta+\alpha v(\bm{x})\leq \mathbb{E}^{\infty}[v(\bm{\phi}_{\pi}^{\bm{x}}(1))], & \forall \bm{x}\in \mathcal{X},\\
   v(\bm{x})\leq 1, &\forall \bm{x}\in \widetilde{\mathcal{X}}\setminus \mathcal{X},
   \end{cases}
   \end{equation}
then $\mathbb{P}^{\infty}\big(\exists k\in \mathbb{N}_{\leq N}. \bm{\phi}^{\bm{x}_0}_{\pi}(k) \in \widetilde{\mathcal{X}}\setminus \mathcal{X} \mid  \bm{x}_0\in \mathcal{X} \big)=\mathbb{P}^{\infty}\big(\widetilde{\bm{\phi}}^{\bm{x}_0}_{\pi}(N) \in \widetilde{\mathcal{X}}\setminus \mathcal{X} \mid  \bm{x}_0\in \mathcal{X} \big)\geq$ \[
   \begin{cases}
   \frac{(\alpha^{N+1}v(\bm{x}_0)-M)(\alpha-1)+\beta(\alpha^{N+1}-1)}{(\alpha+\beta-1)(\alpha^{N+1}-1)}, & \text{if~}\alpha>1,\\
   1+\frac{v(\bm{x}_0)-M}{\beta(N+1)}, &  \text{if~}\alpha=1. 
   \end{cases}
   \]
    \end{theorem}
           \begin{proof}
           1) We first prove the case of $\alpha>1$. 
           
Since  $\mathbb{E}^{\infty}[v(\widetilde{\bm{\phi}}_{\pi}^{\bm{x}}(1))]=v(\bm{x})$ for $\bm{x}\in \widetilde{\mathcal{X}}\setminus \mathcal{X}$,  we conclude that if $v(\bm{x})$ satisfies \eqref{ebr_e7},  it will satisfy 
   \begin{equation*}
    (\alpha-1+\beta)1_{\widetilde{\mathcal{X}}\setminus \mathcal{X}}(\bm{x})+ \mathbb{E}^{\infty}[v(\widetilde{\bm{\phi}}_{\pi}^{\bm{x}}(1))]\geq \alpha v(\bm{x})+\beta,\forall \bm{x}\in \widetilde{\mathcal{X}}.
   \end{equation*}
   
Consequently, for $\bm{x}_0\in \mathcal{X}$, we have 
\begin{equation*}
    \begin{split}
    &\mathbb{E}^{\infty}[v(\widetilde{\bm{\phi}}_{\pi}^{\bm{x}_0}(1))]-\alpha v(\bm{x}_0)\geq \beta+(1-\alpha-\beta)1_{\widetilde{\mathcal{X}}\setminus \mathcal{X}}(\bm{x}_0),\\
    &\mathbb{E}^{\infty}[v(\widetilde{\bm{\phi}}_{\pi}^{\bm{x}_0}(2))]-\alpha\mathbb{E}^{\infty}[v(\widetilde{\bm{\phi}}_{\pi}^{\bm{x}_0}(1))]\geq \beta+(1-\alpha-\beta)\mathbb{E}^{\infty}[1_{\widetilde{\mathcal{X}}\setminus \mathcal{X}}(\widetilde{\bm{\phi}}_{\pi}^{\bm{x}_0}(1))],\\
    &\ldots,\\
    &\mathbb{E}^{\infty}[v(\widetilde{\bm{\phi}}_{\pi}^{\bm{x}_0}(N+1))]-\alpha\mathbb{E}^{\infty}[v(\widetilde{\bm{\phi}}_{\pi}^{\bm{x}_0}(N))]\geq \beta+(1-\alpha-\beta)\mathbb{E}^{\infty}[1_{\widetilde{\mathcal{X}}\setminus \mathcal{X}}(\widetilde{\bm{\phi}}_{\pi}^{\bm{x}_0}(N))].
    \end{split}
\end{equation*}
Thus, we obtain 
\[
\begin{split}
&\mathbb{E}^{\infty}[v(\widetilde{\bm{\phi}}_{\pi}^{\bm{x}_0}(N+1))]-\alpha^{N+1}v(\bm{x}_0)\\
&\geq \beta \sum_{i=0}^N \alpha^i+(1-\alpha-\beta) \sum_{i=0}^N \alpha^{N-i}\mathbb{E}^{\infty}[1_{\widetilde{\mathcal{X}}\setminus \mathcal{X}}(\widetilde{\bm{\phi}}_{\pi}^{\bm{x}_0}(i))]\\
&\geq \beta \frac{1-\alpha^{N+1}}{1-\alpha}+(1-\alpha-\beta) \frac{1-\alpha^{N+1}}{1-\alpha}\times \mathbb{P}^{\infty}\big(\widetilde{\bm{\phi}}^{\bm{x}_0}_{\pi}(N) \in \widetilde{\mathcal{X}}\setminus \mathcal{X} \mid  \bm{x}_0\in \mathcal{X} \big)
\end{split}.\]
The last inequality is obtained according to Lemma \ref{lema1}, which states $\mathbb{P}^{\infty}\big(\widetilde{\bm{\phi}}^{\bm{x}_0}_{\pi}(N) \in \widetilde{\mathcal{X}}\setminus \mathcal{X} \mid  \bm{x}_0\in \mathcal{X} \big)\geq \mathbb{P}^{\infty}\big(\widetilde{\bm{\phi}}^{\bm{x}_0}_{\pi}(i) \in \widetilde{\mathcal{X}}\setminus \mathcal{X} \mid  \bm{x}_0\in \mathcal{X} \big)$ for $i\leq N$, and the fact that $1-\alpha-\beta<0$.

Consequently,
\[
\begin{split}
&\mathbb{P}^{\infty}\big(\widetilde{\bm{\phi}}^{\bm{x}_0}_{\pi}(N) \in \widetilde{\mathcal{X}}\setminus \mathcal{X} \mid  \bm{x}_0\in \mathcal{X} \big)\geq \frac{(\alpha^{N+1}v(\bm{x}_0)-M)(\alpha-1)+\beta(\alpha^{N+1}-1)}{(\alpha+\beta-1)(\alpha^{N+1}-1)}
\end{split}.
\]

2) The conclusion for the case of $\alpha=1$ can be justified via
following the proof of the above one.

The proof is completed.
\end{proof}

It is worth noting here that if there exists a bounded function $v(\bm{x})\colon \widetilde{\mathcal{X}}\rightarrow \mathbb{R}$ satisfying \eqref{ebr_e7} with  $\alpha=1$, then system \eqref{system}, starting from $\bm{x}_0\in \mathcal{X}$, will exit the safe set $\mathcal{X}$ eventually with the probability of one, since $\lim_{N\rightarrow +\infty} 1+\frac{v(\bm{x}_0)-M}{\beta(N+1)}=1$. Therefore, if system \eqref{system} does not feature this property, $\alpha=1$ cannot be used to perform computations.  

\begin{remark}
    When $\alpha=1$, the constraint $v(\bm{x})\leq 1, \forall \bm{x}\in \widetilde{\mathcal{X}}\setminus \mathcal{X}$ in \eqref{ebr_e7} is redundant and can be removed.  
\end{remark}

\begin{remark}
One may wonder whether a sufficient condition for lower bounding the probability of exiting the safe set $\mathcal{X}$ over the time horizon $[0,N]$ in Problem \ref{safe} can be constructed by directly reversing the sign in constraint \eqref{ebr_2} in Theorem \ref{pro_e2}. We will give a brief explanation here that the condition constructed in this manner will consistently yield zero as lower bounds.  

    If there exists a function $v(\bm{x})\colon\widetilde{\mathcal{X}}\rightarrow \mathbb{R}$ such that
    \begin{equation}
   \label{ebr_2lower}
   \begin{cases}
   \mathbb{E}^{\infty}[v(\bm{\phi}_{\pi}^{\bm{x}}(1))]\geq \frac{v(\bm{x})}{\alpha} +\beta, & \forall \bm{x}\in \mathcal{X},\\
   v(\bm{x})\leq 1, & \forall \bm{x}\in  \widetilde{\mathcal{X}}\setminus \mathcal{X},\\
   v(\bm{x})\leq 0, &\forall \bm{x}\in \mathcal{X},   
   \end{cases}
   \end{equation}
 then $\mathbb{P}^{\infty}\big(\exists k\in \mathbb{N}_{\leq N}. \bm{\phi}^{\bm{x}_0}_{\pi}(k) \in \widetilde{\mathcal{X}}\setminus \mathcal{X} \mid  \bm{x}_0\in \mathcal{X} \big)=\mathbb{P}^{\infty}\big( \widetilde{\bm{\phi}}^{\bm{x}_0}_{\pi}(N) \in \widetilde{\mathcal{X}}\setminus \mathcal{X} \mid  \bm{x}_0\in \mathcal{X} \big)\geq$
   \[
    \begin{cases}
       v(\bm{x}_0)+\beta N, &\text{if~}\alpha=1 \wedge \gamma \in (-\infty, 0),\\
       v(\bm{x}_0) \alpha^{-N}+\frac{(1-\alpha^{-N})\alpha \beta}{\alpha-1}, & \text{if~}\alpha\in (0,1)\wedge \gamma \in (-\infty, 0),\\
       1-(1-v(\bm{x}_0))\alpha^{-N},                &\text{if~}\alpha\in (0,1] \wedge \gamma\in [0,\infty),
 \end{cases}
 \] 
 where $\gamma=\beta \alpha -(\alpha-1)$.  This conclusion can be justified by following the proof of Theorem \ref{pro_e2}. However, via $\gamma=\beta \alpha -(\alpha-1)<0$ and $\alpha \in (0,1]$, we have $\beta<0$. Also, 
 since $v(\bm{x})\leq 0$ for $\bm{x}\in \mathcal{X}$, $v(\bm{x}_0)\leq 0$ holds. Thus,  $v(\bm{x}_0)+\beta N\leq 0$, $v(\bm{x}_0) \alpha^{-N}+\frac{(1-\alpha^{-N})\alpha \beta}{\alpha-1}\leq 0$, and $1-(1-v(\bm{x}_0))\alpha^{-N}\leq 0$ hold. 

\end{remark}

\section{Finite-time Reach-avoid Verification}
\label{sec:comp1}
In this section, we present our sufficient conditions for characterizing upper and lower bounds of the reach-avoid probability in Problem \ref{ravoid}. The sufficient condition for upper bounds is formulated in Subsection \ref{sub:ubfr} and the one for lower bounds is introduced in Subsection \ref{sub:lbfr}.

Similar to \cite{xue2021reach}, we define a switched system, which facilitates the transformation of the reach-avoid problem in Problem \ref{ravoid} to a mere reachability problem. The switched system is constructed by freezing the dynamics of system \eqref{system} upon either exiting the safe set $\mathcal{X}$ or reaching the target set $\mathcal{X}_r$. 

\begin{definition}
\label{switch}
The switched stochastic discrete-time system, which is built upon system \eqref{system}, is a quadruple $(\widehat{\mathcal{L}},\widehat{\mathcal{X}},\bm{x}_0, \widehat{\bm{f}})$ with the following components:
\begin{enumerate}
\item[-] $\widehat{\mathcal{L}}=\{1,2,3\}$ is a set of three locations;
\item[-] $\widehat{\mathcal{X}}\subseteq \mathbb{R}^n$ is the state constraint set;
\item[-] $\bm{x}_0\in \widehat{\mathcal{X}}$ is the initial state;
\item[-] $\widehat{\bm{f}}(\cdot,\cdot)\colon\mathbb{R}^n \times \Theta \rightarrow \mathbb{R}^n$, where \[\widehat{\bm{f}}(\bm{x},\bm{\theta})=\sum_{i=1}^3 1_{\widehat{\mathcal{X}}_i}(\bm{x})\widehat{\bm{f}}_i(\bm{x},\bm{\theta})\] with $\widehat{\bm{f}}_1(\bm{x},\bm{\theta})=\bm{f}(\bm{x},\bm{\theta})$, $\widehat{\bm{f}}_2(\bm{x},\bm{\theta})=\bm{x}$ and  $\widehat{\bm{f}}_3(\bm{x},\bm{\theta})=\bm{x}$, and $1_{\widehat{\mathcal{X}}_i}(\bm{x})$ is the indicator function of the set $\widehat{\mathcal{X}}_i$, i.e., $1_{\widehat{\mathcal{X}}_i}(\bm{x})=1$ if $\bm{x}\in \widehat{\mathcal{X}}_i$; otherwise, $1_{\widehat{\mathcal{X}}_i}(\bm{x})=0$,
\end{enumerate}
where 
\begin{enumerate}
\item $\widehat{\mathcal{X}}$ is a set satisfying $\widehat{\Omega}\subset\widehat{\mathcal{X}}$, where
\[\widehat{\Omega}=\{\bm{x}'\in \mathbb{R}^n\mid \bm{x}'=\bm{f}(\bm{x},\bm{\theta}), \bm{x}\in \mathcal{X}, \bm{\theta}\in \Theta\}\cup \mathcal{X};\]
\item $\widehat{\mathcal{X}}_1=\mathcal{X}\setminus \mathcal{X}_r$;
\item $\widehat{\mathcal{X}}_2=\mathcal{X}_r$;
\item $\widehat{\mathcal{X}}_3=\widehat{\mathcal{X}}\setminus \mathcal{X}$.
\end{enumerate}
The evolution of the state of this switched system is governed by the iterative map
\begin{equation}
\label{sdss2}
\begin{split}
&\bm{x}(l+1)=\widehat{\bm{f}}(\bm{x}(l),\bm{\theta}(l)), \forall l\in \mathbb{N},\\
&\bm{x}(0)=\bm{x}_0\in \widehat{\mathcal{X}}.
\end{split}
\end{equation}
\end{definition}

The trajectory of system \eqref{sdss2}, induced by initial state $\bm{x}_0\in \widehat{\mathcal{X}}$ and disturbance policy $\pi$, is denoted by 
$\widehat{\bm{\phi}}_{\pi}^{\bm{x}_0}(\cdot)\colon \mathbb{N} \rightarrow \mathbb{R}^n$.

  When system \eqref{system}, initialized at $\bm{x}_0 \in \mathcal{X}\setminus \mathcal{X}_r$, enters the region $\mathcal{X}_r$ for the first time at time $i\in \mathbb{N}$ without leaving the safe set $\mathcal{X}$ beforehand, system \eqref{sdss2} starting from $\bm{x}_0 \in \mathcal{X}\setminus \mathcal{X}_r$ also experiences its initial entry at this time, and vice versa. Further, when entering $\mathcal{X}_r$ and $\widehat{\mathcal{X}}\setminus \mathcal{X}$, system \eqref{sdss2} will remain confined indefinitely. Thus, the set $\widehat{\mathcal{X}}$ is a robust invariant for system \eqref{sdss2} \cite{xue2021reach}, that is, trajectories of system \eqref{sdss2} originating from the set $\widehat{\mathcal{X}}$ will never leave it. Consequently, the probability of system \eqref{sdss2} entering $\mathcal{X}_r$ at time $t=i$ aligns with the cumulative probability of system \eqref{system} entering this target set up to time $i$ without exiting $\mathcal{X}$ until the target hitting event occurs.

\begin{lemma}
\label{lemma2}
For $i\in \mathbb{N}$, $\mathbb{P}^{\infty}(\exists k\in \mathbb{N}_{\leq i}. \bm{\phi}_{\pi}^{\bm{x}_0}(k)\in \mathcal{X}_r
\bigwedge \forall l\in \mathbb{N}_{\leq k}. \bm{\phi}_{\pi}^{\bm{x}_0}(l)\in \mathcal{X} \mid \bm{x}_0 \in \mathcal{X}\setminus \mathcal{X}_r)=\mathbb{P}^{\infty}(\widehat{\bm{\phi}}_{\pi}^{\bm{x}_0}(i)\in \mathcal{X}_r \mid \bm{x}_0 \in \mathcal{X}\setminus \mathcal{X}_r)$. Thus, $\mathbb{P}^{\infty}(\widehat{\bm{\phi}}_{\pi}^{\bm{x}_0}(i)\in \mathcal{X}_r  \mid \bm{x}_0 \in \mathcal{X}\setminus \mathcal{X}_r)\leq \mathbb{P}^{\infty}(\widehat{\bm{\phi}}_{\pi}^{\bm{x}_0}(j)\in \mathcal{X}_r  \mid \bm{x}_0 \in \mathcal{X}\setminus \mathcal{X}_r)$ if $i\leq j$.
\end{lemma}
\subsection{Sufficient Conditions for Upper Bounds}
\label{sub:ubfr}
In this subsection, we introduce our sufficient condition for upper bounding the probability in Problem \ref{ravoid}.

The sufficient condition requires the existence of a non-negative function $v(\bm{x})$ defined on the set $\widehat{\mathcal{X}}$ that satisfies two key properties: (1) it must be greater than or equal to one on the target set $\mathcal{X}_r$, and (2) its $\alpha$-scaled expected value at the next step under the dynamics in \eqref{system} should increase by at most $\alpha \beta$ relative to its current value for $\bm{x}\in \mathcal{X}\setminus \mathcal{X}_r$, where $\alpha \in (0,1]$, and $\beta \in [0,1]$. Upper bounds are derived as follows. Lemma 2 equates the probability of system \eqref{sdss2} entering $\mathcal{X}_r$ at $t=N$ with the cumulative probability of system \eqref{system} reaching $\mathcal{X}_r$ by time $N$ without leaving $\mathcal{X}$ before hitting the target. Then, we derive upper bounds using system \eqref{sdss2} by analyzing $\alpha$. Since system \eqref{sdss2} remains stationary when starting from $\widehat{\mathcal{X}}\setminus \mathcal{X}$ and $\mathcal{X}_r$, implying $\mathbb{E}^{\infty}[v(\widehat{\bm{\phi}}_{\pi}^{\bm{x}}(1))]=v(\bm{x})$ for $\bm{x}\in (\widehat{\mathcal{X}}\setminus \mathcal{X})\cup \mathcal{X}_r$, we reformulate the constraint on $v(\bm{x})$ using system \eqref{sdss2}: its $\alpha$-scaled expected value at the next step along \eqref{sdss2} grows by at most $\alpha \beta$ relative to its current value across the entire set $\widehat{\mathcal{X}}$. For $\alpha\in (0,1)$ and $\alpha=1$, the upper bounds follow from recursively applying this reformulated constraint.
 
\begin{theorem}
    \label{pro_e4}
    If there exist a function $v(\bm{x})\colon\widehat{\mathcal{X}}\rightarrow \mathbb{R}$, $\alpha \in (0,1]$, and $\beta \in [0,1]$ such that 
    \begin{equation}
   \label{ebr_4}
   \begin{cases}
   \mathbb{E}^{\infty}[v(\bm{\phi}_{\pi}^{\bm{x}}(1))]\leq \frac{v(\bm{x})}{\alpha} +\beta, & \forall \bm{x}\in \mathcal{X}\setminus \mathcal{X}_r,\\
   v(\bm{x})\geq 1, & \forall \bm{x}\in \mathcal{X}_r,\\
    v(\bm{x})\geq 0, & \forall \bm{x}\in  \widehat{\mathcal{X}}\setminus \mathcal{X}_r,
   \end{cases}
   \end{equation}
 then $\mathbb{P}^{\infty}\big(\exists k\in \mathbb{N}_{\leq N}. \bm{\phi}^{\bm{x}_0}_{\pi}(k) \in \mathcal{X}_r\bigwedge \forall l\in \mathbb{N}_{\leq k}. \bm{\phi}^{\bm{x}_0}_{\pi}(l)\in \mathcal{X} \mid  \bm{x}_0 \in \mathcal{X}\setminus \mathcal{X}_r \big)=\mathbb{P}^{\infty}\big(\widehat{\bm{\phi}}^{\bm{x}_0}_{\pi}(N) \in \mathcal{X}_r \mid  \bm{x}_0 \in \mathcal{X}\setminus \mathcal{X}_r \big)\leq$
   \[
    \begin{cases}
       v(\bm{x}_0)+\beta N, &\text{if~}\alpha=1 \wedge \beta \in [0,1],\\
       v(\bm{x}_0) \alpha^{-N}+\frac{(1-\alpha^{-N})\alpha \beta}{\alpha-1}, & \text{if~}\alpha\in (0,1)\wedge \beta \in [0,1].
 \end{cases}
 \]  
    \end{theorem}
    \begin{proof}
    1) We first justify the case of $\alpha\in (0,1)\wedge \beta \in [0,1]$.
    
    Since $\mathbb{E}^{\infty}[v(\widehat{\bm{\phi}}_{\pi}^{\bm{x}}(1))]=v(\bm{x})$ for $\bm{x}\in (\widehat{\mathcal{X}}\setminus \mathcal{X}) \cup \mathcal{X}_r$, we can obtain 
\[\mathbb{E}^{\infty}[v(\widehat{\bm{\phi}}_{\pi}^{\bm{x}}(1))]-\frac{1}{\alpha}v(\bm{x})=(1-\frac{1}{\alpha}) v(\bm{x}), \forall \bm{x}\in (\widehat{\mathcal{X}}\setminus \mathcal{X}) \cup \mathcal{X}_r.\]
Therefore, \[\mathbb{E}^{\infty}[v(\widehat{\bm{\phi}}_{\pi}^{\bm{x}}(1))]-\frac{1}{\alpha}v(\bm{x})\leq 0 \leq \beta, \forall \bm{x}\in (\widehat{\mathcal{X}}\setminus \mathcal{X}) \cup \mathcal{X}_r,\] which can be justified based on the facts that $v(\bm{x})\geq 0$ for $\bm{x}\in \widehat{\mathcal{X}}$, $\alpha\in (0,1)$ and $\beta \in [0,1]$. Therefore, if $v(\bm{x})$ satisfies \eqref{ebr_4}, it will satisfy 
     \begin{equation}
   \label{ebs_3}
   \begin{cases}
   \mathbb{E}^{\infty}[v(\widehat{\bm{\phi}}_{\pi}^{\bm{x}}(1))]\leq \frac{v(\bm{x})}{\alpha} +\beta, & \forall \bm{x}\in \widehat{\mathcal{X}},\\
   v(\bm{x})\geq 1_{\mathcal{X}_r}(\bm{x}), & \forall \bm{x}\in \widehat{\mathcal{X}}.
   \end{cases}
   \end{equation}

    According to \eqref{ebs_3}, we have  
    \[
    \begin{split}
    &v(\bm{x}_0)\geq 1_{\mathcal{X}_r}(\bm{x}_0),\\
    &\alpha^{-1}v(\bm{x}_0)+\beta\geq \mathbb{E}^{\infty}[v(\widehat{\bm{\phi}}_{\pi}^{\bm{x}_0}(1))]\geq \mathbb{E}^{\infty}[1_{\mathcal{X}_r}(\widehat{\bm{\phi}}_{\pi}^{\bm{x}_0}(1))],\\
    &\ldots,\\
    &\alpha^{-N}v(\bm{x}_0)+\beta \sum_{i=0}^{N-1} \alpha^{-i} \geq  \mathbb{E}^{\infty}[v(\widehat{\bm{\phi}}_{\pi}^{\bm{x}_0}(N))]\geq \mathbb{E}^{\infty}[1_{\mathcal{X}_r}(\widehat{\bm{\phi}}_{\pi}^{\bm{x}_0}(N))].
    \end{split}
    \]
    Thus, 
 \[
 \begin{split}
 &\mathbb{P}^{\infty}\big(\widehat{\bm{\phi}}^{\bm{x}_0}_{\pi}(N) \in \mathcal{X}_r \mid  \bm{x}_0 \in \mathcal{X}\setminus \mathcal{X}_r \big)=\mathbb{E}^{\infty}[1_{\mathcal{X}_r}(\widehat{\bm{\phi}}_{\pi}^{\bm{x}_0}(N))]\\
 &\leq \alpha^{-N}v(\bm{x}_0)+\alpha \beta \frac{(1-\alpha^{-N})}{\alpha-1}.
 \end{split}
 \]According to Lemma \ref{lemma2}, we have the conclusion.

    2) The conclusion on the case of $\alpha=1\wedge \beta \in [0,1]$ can be justified by following the proof for the above one.

The proof is completed.
   \end{proof}

Similar to the analysis for the parameter $\beta$ in Theorem \ref{pro_e2}, the parameter $\beta$ in Theorem \ref{pro_e4} is also required to be less than or equal to 1, which is subject to a practical restriction. Assigning a value greater than one to $\beta$ would result in bounds exceeding unity when $N \geq 1$, rendering them ineffective. 

On the other hand, the analysis for the upper bounds in Theorem \ref{pro_e4} is similar to the one in Theorem \ref{pro_e2}. When $\alpha=1$ and $\beta=0$, the upper bound is $v(\bm{x}_0)$. It is an upper bound of the probability $\mathbb{P}^{\infty}\big(\exists k\in \mathbb{N}. \bm{\phi}^{\bm{x}_0}_{\pi}(k) \in \mathcal{X}_r \wedge \forall l\in \mathbb{N}_{\leq k}. \bm{\phi}^{\bm{x}_0}_{\pi}(k) \in \mathcal{X} \mid  \bm{x}_0 \in \mathcal{X}\setminus \mathcal{X}_r \big)$, which is the probability of reaching the target set $\mathcal{X}_r$ eventually while staying within the safe set $\mathcal{X}$ before the first target hitting time.

\begin{remark}
\label{comparionwithsan}
In Theorem \ref{pro_e4}, $\alpha$ is limited to be in $(0,1]$. We can obtain upper bounds for the case with $\alpha>1$ according to Theorem 3 in Chapter 3 in \cite{kushner1967stochastic} or Proposition \ref{theorem3} in Section \ref{sec:comp}, as shown in Corollary \ref{pro_e41}. 
\begin{corollary}
    \label{pro_e41}
    If there exist a continuous non-negative function $v(\bm{x})\colon\widehat{\mathcal{X}}\rightarrow \mathbb{R}$, $\alpha \in [1,\infty)$, and $\beta \in [0,1]$ such that 
    \begin{equation}
   \label{ebr_41}
   \begin{cases}
   v(\bm{x}_0)<1,\\
   \mathbb{E}^{\infty}[v(\bm{\phi}_{\pi}^{\bm{x}}(1))]\leq \frac{v(\bm{x})}{\alpha} +\beta, & \forall \bm{x}\in \mathcal{X}\setminus \mathcal{X}_r,\\
   (1-\frac{1}{\alpha})v(\bm{x})-\beta\leq 0, &\forall \bm{x} \in \widehat{\mathcal{X}}\setminus \mathcal{X},\\
   v(\bm{x})\geq 1, & \forall \bm{x}\in \mathcal{X}_r, \\
    v(\bm{x})\geq 0, & \forall \bm{x}\in \widehat{\mathcal{X}},
   \end{cases}
   \end{equation}
 then $\mathbb{P}^{\infty}\big(\exists k\in \mathbb{N}_{\leq N}. \bm{\phi}^{\bm{x}_0}_{\pi}(k) \in \mathcal{X}_r\bigwedge \forall l\in \mathbb{N}_{\leq k}. \bm{\phi}^{\bm{x}_0}_{\pi}(l)\in \mathcal{X} \mid  \bm{x}_0 \in \mathcal{X}\setminus \mathcal{X}_r \big)\leq$
   \begin{equation}
   \label{upper}
    \begin{cases}
       v(\bm{x}_0)+\beta N, &\text{if~}\alpha=1, \\
       v(\bm{x}_0) \alpha^{-N}+\frac{(1-\alpha^{-N})\alpha \beta}{\alpha-1}, & \text{if~}\alpha>1 \wedge \frac{\beta \alpha}{\alpha-1}>1,\\
       1-(1-v(\bm{x}_0)) (1-\beta)^{N}, & \text{if~} \alpha>1 \wedge \frac{\beta \alpha}{\alpha-1}\leq 1.
 \end{cases}
 \end{equation}
    \end{corollary}    
    \begin{proof}
Like in \cite{santoyo2021barrier}, consider the stopped version of the stochastic process satisfying \eqref{system}, which ceases evolving upon exiting the set $\mathcal{X}$.  Given a disturbance signal $\pi$, the trajectory of the stopped process, starting from $\bm{x}_0$, is denoted by $\{\bar{\bm{\phi}}_{\pi}^{\bm{x}_0}(i)\}_{i\in \mathbb{N}}$, which satisfies 
\begin{equation*}
\begin{cases}
        \bar{\bm{\phi}}_{\pi}^{\bm{x}_0}(i+1)=1_{\widehat{\mathcal{X}}\setminus \mathcal{X}}(\bar{\bm{\phi}}_{\pi}^{\bm{x}_0}(i))+1_{\mathcal{X}}(\bar{\bm{\phi}}_{\pi}^{\bm{x}_0}(i))\bm{f}(\bar{\bm{\phi}}_{\pi}^{\bm{x}_0}(i),\bm{\theta}(i)),\\
    \bar{\bm{\phi}}_{\pi}^{\bm{x}_0}(0)=\bm{x}_0.
\end{cases}
\end{equation*}
It is easy to conclude that any sample trajectory for the stopped process starting from any state in $\widehat{\mathcal{X}}$ cannot leave the set $\widehat{\mathcal{X}}$, and $\{\pi\mid \exists k\in \mathbb{N}_{\leq N}. \bar{\bm{\phi}}^{\bm{x}_0}_{\pi}(k) \in \mathcal{X}_r\}$=$\{\pi\mid \exists k\in \mathbb{N}_{\leq N}. \bm{\phi}^{\bm{x}_0}_{\pi}(k) \in \mathcal{X}_r\bigwedge \forall l\in \mathbb{N}_{\leq k}. \bm{\phi}^{\bm{x}_0}_{\pi}(l)\in \mathcal{X} \}$ for $\bm{x}_0\in \widehat{\mathcal{X}}\setminus \mathcal{X}$. Therefore, the probability $\mathbb{P}^{\infty}\big(\exists k\in \mathbb{N}_{\leq N}. \bar{\bm{\phi}}^{\bm{x}_0}_{\pi}(k) \in \mathcal{X}_r \mid \bm{x}_0 \in \mathcal{X}\setminus \mathcal{X}_r \big)$ of reaching $\mathcal{X}_r$ for the stopped process  is equivalent to the reach-avoid probability $\mathbb{P}^{\infty}\big(\exists k\in \mathbb{N}_{\leq N}. \bm{\phi}^{\bm{x}_0}_{\pi}(k) \in \mathcal{X}_r\bigwedge \forall l\in \mathbb{N}_{\leq k}. \bm{\phi}^{\bm{x}_0}_{\pi}(l)\in \mathcal{X} \mid \bm{x}_0 \in \mathcal{X}\setminus \mathcal{X}_r \big)$. 

According to Theorem 3 in Chapter 3 in \cite{kushner1967stochastic} (or, Proposition \ref{theorem3} in Section \ref{sec:comp}), we have that if there exist a continuous non-negative function $v(\bm{x})\colon\widehat{\mathcal{X}}\rightarrow \mathbb{R}$, $\alpha \in [1,\infty)$, and $\beta \in [0,1]$ such that
    \begin{equation}
    \label{14}
   \begin{cases}
   v(\bm{x}_0)<1,\\
   \mathbb{E}^{\infty}[v(\bar{\bm{\phi}}_{\pi}^{\bm{x}}(1))]\leq \frac{v(\bm{x})}{\alpha} +\beta, & \forall \bm{x}\in \widehat{\mathcal{X}}\setminus \mathcal{X}_r,\\
   v(\bm{x})\geq 1, & \forall \bm{x}\in \mathcal{X}_r, \\
    v(\bm{x})\geq 0, & \forall \bm{x}\in \widehat{\mathcal{X}},
   \end{cases}
   \end{equation}
   then $\mathbb{P}^{\infty}\big(\exists k\in \mathbb{N}_{\leq N}. \bar{\bm{\phi}}^{\bm{x}_0}_{\pi}(k) \in \mathcal{X}_r \big)$ has upper bounds in \eqref{upper}, and thus $\mathbb{P}^{\infty}\big(\exists k\in \mathbb{N}_{\leq N}. \bm{\phi}^{\bm{x}_0}_{\pi}(k) \in \mathcal{X}_r\bigwedge \forall l\in \mathbb{N}_{\leq k}. \bm{\phi}^{\bm{x}_0}_{\pi}(l)\in \mathcal{X} \mid \bm{x}_0 \in \mathcal{X}\setminus \mathcal{X}_r \big)$ has upper bounds in \eqref{upper}. 
 
  On the other hand, since $\bar{\bm{\phi}}_{\pi}^{\bm{x}}(1)=\bm{x}$, $\forall \bm{x}\in \widehat{\mathcal{X}}\setminus \mathcal{X}, \forall \pi$, and $\bar{\bm{\phi}}_{\pi}^{\bm{x}}(1)=\bm{\phi}_{\pi}^{\bm{x}}(1)$, $\forall \bm{x}\in \mathcal{X}, \forall \pi$, we obtain that constraint \eqref{14} is equivalent to \eqref{ebr_41}. Thus, the conclusion holds and the proof is completed.
   \end{proof}

Following the comparison between constraints \eqref{ebr_2} and \eqref{ebr_33}, we conclude that if there exist a function $v(\bm{x})$, $\alpha \geq 1$,  and $\beta \in [0, 1]$ satisfying \eqref{ebr_41}, then the same function $v(\bm{x})$ with $\frac{1}{\alpha}$ and $\beta \in [0, 1]$ will satisfy \eqref{ebr_4}. Additionally, it is noted that a similar constraint, associated with the same upper bounds in \eqref{upper}, has been presented in Proposition 2 of \cite{santoyo2021barrier}, which is also derived from Theorem 3 in Chapter 3 of \cite{kushner1967stochastic} (or, Proposition \ref{theorem3}). \cite{santoyo2021barrier} utilized this constraint to determine upper bounds of the probability of reaching the set $\mathcal{X}_r$ for stopped processes that cease evolving upon exiting the interior of the set $\mathcal{X}$. The stopped process is the same as the one in the proof of Corollary \ref{pro_e41} when the set $\mathcal{X}$ is open. As shown in the proof of Corollary \ref{pro_e41}, the probability of reaching $\mathcal{X}_r$ for the stopped process in \cite{santoyo2021barrier} is equivalent to the reach-avoid probability $\mathbb{P}^{\infty}\big(\exists k\in \mathbb{N}_{\leq N}. \bm{\phi}^{\bm{x}_0}_{\pi}(k) \in \mathcal{X}_r \bigwedge \forall l\in \mathbb{N}_{\leq k}. \bm{\phi}^{\bm{x}_0}_{\pi}(l)\in \mathcal{X} \mid \bm{x}_0 \in \mathcal{X}\setminus \mathcal{X}_r \big)$ when the set $\mathcal{X}$ is open. This indicates that the safety verification problem in \cite{santoyo2021barrier} for stopped processes is actually a reach-avoid verification problem for the original process in the present work. Although the constraint in \cite{santoyo2021barrier} and the one \eqref{ebr_41} are constructed based on the same stopped process (when the set $\mathcal{X}$ is open) and Theorem 3 in Chapter 3 of \cite{kushner1967stochastic}, constraint \eqref{ebr_41} is more stringent than the one in \cite{santoyo2021barrier}, including additional conditions of $(1-\frac{1}{\alpha})v(\bm{x})-\beta\leq 0, \forall \bm{x} \in \widehat{\mathcal{X}}\setminus \mathcal{X}$ and $v(\bm{x})\geq 0, \forall \bm{x}\in \widehat{\mathcal{X}}\setminus \mathcal{X}$. To the best of my knowledge, I believe that these additional constraints are necessary and cannot be omitted according to Theorem 3 in Chapter 3 of \cite{kushner1967stochastic} (or, Proposition \ref{theorem3}).
\end{remark}

\begin{remark}
 Under the assumption that $\bm{f}(\cdot,\cdot)\colon \mathcal{X}\times \Theta \rightarrow \mathcal{X}$ (i.e., $\mathcal{X}$ is a robust invariant set \footnote{Generally, a robust invariant set is fundamentally challenging, even impossible to compute even if it exists.} for system \eqref{system}), the $c$-martingale was employed in \cite{jagtap2018temporal} to certify upper bounds of the probability with which system \eqref{system}  starting from $\bm{x}_0\in \mathcal{X}\setminus \mathcal{X}_r$ will reach the set $\mathcal{X}_r$ within a specified bounded time horizon. Under this strong assumption, $c$-martingale satisfies \eqref{ebr_4} in Theorem \ref{pro_e4} with $\alpha=1$. This assumption is also imposed in \cite{zhi2024unifying}. It is worth noting that $\widehat{\mathcal{X}}$ in our conditions is not needed if $\bm{f}(\cdot,\cdot)\colon \mathcal{X}\times \Theta \rightarrow \mathcal{X}$ holds.
\end{remark}

\subsection{Sufficient Conditions for Lower Bounds}
\label{sub:lbfr}
In this subsection, we introduce our sufficient condition for lower bounding the probability in Definition \ref{ravoid}.

The sufficient condition requires a function $v(\bm{x}): \widehat{\mathcal{X}}\rightarrow \mathbb{R}$ satisfying four key properties: (1) its admits a finite upper bound over $\widehat{\mathcal{X}}$, (2) it is less than or equal to one on the target set $\mathcal{X}_r$, (3) its expected value at the next step under the dynamics in \eqref{system} exceeds its $\alpha$-scaled value by at least $\beta$ for $\bm{x}\in \mathcal{X}\setminus \mathcal{X}_r$, and (4) it is less than or equal to $-\frac{\beta}{\alpha-1}$ over the set $\widehat{\mathcal{X}}\setminus \mathcal{X}$, where $\alpha \in (1,\infty)$ and $\beta\in (1-\alpha,\infty)$. Lower bounds are derived as follows. Lemma 2 equates the probability of system \eqref{sdss2} entering $\mathcal{X}_r$ at $t=N$ with the cumulative probability of system \eqref{system} reaching $\mathcal{X}_r$ by time $N$ without leaving $\mathcal{X}$ before hitting the target. Then, we derive a lower bound using system \eqref{sdss2}. System \eqref{sdss2} remains stationary  for initial states in  $(\widehat{\mathcal{X}}\setminus \mathcal{X}) \cup \mathcal{X}_r$, implying $\mathbb{E}^{\infty}[v(\widehat{\bm{\phi}}_{\pi}^{\bm{x}}(1))]=v(\bm{x})$ for $\bm{x}\in (\widehat{\mathcal{X}}\setminus \mathcal{X})\cup \mathcal{X}_r$. Given that $v(\bm{x})\leq 1$ on $\mathcal{X}_r$ and $v(\bm{x})\leq -\frac{\beta}{\alpha-1}$ over $\widehat{\mathcal{X}}\setminus \mathcal{X}$, we reformulate the constraint over $v(\bm{x})$ using system \eqref{sdss2}: its expected value at the next step should grow by at least $\beta-(\alpha-1+\beta)1_{\widehat{\mathcal{X}}\setminus \mathcal{X}}(\bm{x})$ relative to its $\alpha$-scaled current value $v(\bm{x})$ for $\bm{x}\in \widehat{\mathcal{X}}$. Recursively applying this reformulated constraint, together with the upper bound on $v$, yields a lower bound.

\begin{theorem}
    \label{pro_e6}
      If there exist a function $v(\bm{x})\colon\widehat{\mathcal{X}}\rightarrow \mathbb{R}$ with $\sup_{\bm{x}\in \widehat{\mathcal{X}}}v(\bm{x})\leq M$, $\alpha \in (1,\infty)$, and $\beta\in (1-\alpha,\infty)$ such that 
        \begin{equation}
   \label{ebr_e6}
   \begin{cases}
   \beta+\alpha v(\bm{x})\leq \mathbb{E}^{\infty}[v(\bm{\phi}_{\pi}^{\bm{x}}(1))], & \forall \bm{x}\in \mathcal{X}\setminus \mathcal{X}_r,\\
   v(\bm{x})\leq 1, &\forall \bm{x}\in \mathcal{X}_r,\\
   (\alpha-1)v(\bm{x})\leq -\beta, &\forall \bm{x}\in \widehat{\mathcal{X}}\setminus \mathcal{X},
   \end{cases}
   \end{equation}
 then $\mathbb{P}^{\infty}\big(\exists k\in \mathbb{N}_{\leq N}. \bm{\phi}^{\bm{x}_0}_{\pi}(k) \in \mathcal{X}_r\bigwedge \forall l\in \mathbb{N}_{\leq k}. \bm{\phi}^{\bm{x}_0}_{\pi}(l)\in \mathcal{X} \mid  \bm{x}_0 \in \mathcal{X}\setminus \mathcal{X}_r \big)=\mathbb{P}^{\infty}\big(\widehat{\bm{\phi}}^{\bm{x}_0}_{\pi}(N) \in \mathcal{X}_r \mid  \bm{x}_0 \in \mathcal{X}\setminus \mathcal{X}_r \big)\geq$ 
   \[
   \frac{(\alpha^{N+1}v(\bm{x}_0)-M)(\alpha-1)+\beta(\alpha^{N+1}-1)}{(\alpha+\beta-1)(\alpha^{N+1}-1)}.
   \]
    \end{theorem}
           \begin{proof}
  Since $\mathbb{E}^{\infty}[v(\widehat{\bm{\phi}}_{\pi}^{\bm{x}}(1))]=v(\bm{x})$ for $\bm{x}\in (\widehat{\mathcal{X}}\setminus \mathcal{X}) \cup \mathcal{X}_r$, we can obtain that if $v(\bm{x})$ satisfies \eqref{ebr_e6}, it will satisfy           
   \begin{equation*}
    (\alpha-1+\beta)1_{\mathcal{X}_r}(\bm{x})+\mathbb{E}^{\infty}[v(\widehat{\bm{\phi}}_{\pi}^{\bm{x}}(1))]\geq \alpha v(\bm{x})+\beta,  \forall \bm{x}\in \widehat{\mathcal{X}}.
   \end{equation*}

Thus, for $\bm{x}_0 \in \mathcal{X}\setminus \mathcal{X}_r$, we have 
\begin{equation*}
    \begin{split}
    &\mathbb{E}^{\infty}[v(\widehat{\bm{\phi}}_{\pi}^{\bm{x}_0}(1))]-\alpha v(\bm{x}_0)\geq \beta+(1-\alpha-\beta)1_{\mathcal{X}_r}(\bm{x}_0),\\
    &\mathbb{E}^{\infty}[v(\widehat{\bm{\phi}}_{\pi}^{\bm{x}_0}(2))]-\alpha\mathbb{E}^{\infty}[v(\widehat{\bm{\phi}}_{\pi}^{\bm{x}_0}(1))]\geq \beta+(1-\alpha-\beta)\mathbb{E}^{\infty}[1_{\mathcal{X}_r}(\widehat{\bm{\phi}}_{\pi}^{\bm{x}_0}(1))],\\
    &\ldots,\\
    &\mathbb{E}^{\infty}[v(\widehat{\bm{\phi}}_{\pi}^{\bm{x}_0}(N+1))]-\alpha\mathbb{E}^{\infty}[v(\widehat{\bm{\phi}}_{\pi}^{\bm{x}_0}(N))]\geq \beta+(1-\alpha-\beta)\mathbb{E}^{\infty}[1_{\mathcal{X}_r}(\widehat{\bm{\phi}}_{\pi}^{\bm{x}_0}(N))].
    \end{split}
\end{equation*}
Thus, we can obtain 
\[
\begin{split}
&\mathbb{E}^{\infty}[v(\widehat{\bm{\phi}}_{\pi}^{\bm{x}_0}(N+1))]-\alpha^{N+1}v(\bm{x}_0)\\
&\geq \beta \sum_{i=0}^N \alpha^i+(1-\alpha-\beta) \sum_{i=0}^N \alpha^{N-i}\mathbb{E}^{\infty}[1_{\mathcal{X}_r}(\widehat{\bm{\phi}}_{\pi}^{\bm{x}_0}(i))]\\
&\geq \beta \frac{1-\alpha^{N+1}}{1-\alpha}+(1-\alpha-\beta) \frac{1-\alpha^{N+1}}{1-\alpha} \times \mathbb{P}^{\infty}\big(\widehat{\bm{\phi}}^{\bm{x}_0}_{\pi}(N) \in \mathcal{X}_r \mid  \bm{x}_0 \in \mathcal{X}\setminus \mathcal{X}_r \big)
\end{split}. \]
The last inequality is obtained via Lemma \ref{lemma2}, which states $\mathbb{P}^{\infty}\big(\widehat{\bm{\phi}}^{\bm{x}_0}_{\pi}(N) \in \mathcal{X}_r \mid  \bm{x}_0 \in \mathcal{X}\setminus \mathcal{X}_r \big)\geq \mathbb{P}^{\infty}\big(\widehat{\bm{\phi}}^{\bm{x}_0}_{\pi}(i) \in \mathcal{X}_r \mid  \bm{x}_0 \in \mathcal{X}\setminus \mathcal{X}_r \big)$ for $i\leq N$, and the fact that $1-\alpha-\beta<0$.

Consequently,
\[
\begin{split}
&\mathbb{P}^{\infty}\big(\widehat{\bm{\phi}}^{\bm{x}_0}_{\pi}(N) \in \mathcal{X}_r \mid  \bm{x}_0 \in \mathcal{X}\setminus \mathcal{X}_r \big)\geq \frac{(\alpha^{N+1}v(\bm{x}_0)-M)(\alpha-1)+\beta(\alpha^{N+1}-1)}{(\alpha+\beta-1)(\alpha^{N+1}-1)}
\end{split}.
\]
According to Lemma \ref{lemma2}, we have the conclusion. The proof is completed.
\end{proof}

\begin{remark}
Comparing Theorem \ref{pro_e7} and \ref{pro_e6}, we observe that $\alpha$ cannot be equal to one in Theorem \ref{pro_e6}. Since if $\alpha=1$, we have $\beta>0$ from $\beta \in (1-\alpha,\infty)$. However, we have  $\beta \leq 0$ from $(\alpha-1)v(\bm{x})\leq -\beta, \forall \bm{x}\in \widehat{\mathcal{X}}\setminus \mathcal{X}$. This is a contradiction. 
\end{remark}

\section{Examples}
\label{sec:ex}
In this section, we demonstrate the performance of the proposed conditions for safety and reach-avoid verification on two numerical examples, using the semi-definite programming tool Mosek \cite{aps2019mosek}.

\begin{example}
\label{ex2}
Consider the following one-dimensional discrete-time system:
\begin{equation*}
x(l+1)=x(l)+d(l),
\end{equation*}
where $d(\cdot)\colon\mathbb{N}\rightarrow \Theta=[-0.1,0.1]$, $\mathcal{X}=\{\,x\mid h(x)\leq 0\,\}$ with $h(x)=x^2-1$, $\bm{x}_0=0.2$, $\mathcal{X}_r=\{\,x\mid (x-0.9)^2-10^{-4}\leq 0\,\}$, and $N=30$. Besides, we assume that the probability distribution on $\Theta$ is the uniform distribution. The probabilities in Problem \ref{safe} and \ref{ravoid} obtained via Monte Carlo methods, which used $10^4$ sample paths, are around 0.0085 and 0.0128, respectively.

Given the small exact probabilities in Problems \ref{safe} and \ref{ravoid}, we only estimate their upper bounds in this example. The set $\widetilde{\mathcal{X}} = \widehat{\mathcal{X}} = \{x \mid x^2 - 2 \leq 0\}$ is employed in solving \eqref{ebr_2}, \eqref{ebr_33}, \eqref{ebr_4}, and \eqref{ebr_41}. To address these constraints, we encode them into semi-definite programs with sum of squares decomposition techniques for multivariate polynomials.  Utilizing polynomials $v(\bm{x})$ of varying degrees, the computed upper bounds are summarized in Tables \ref{tab:my_label3} and \ref{tab:my_label4}. Table \ref{tab:my_label3} illustrates that the constraint \eqref{ebr_2} with $\alpha=\frac{1}{1.1}$ can yield tighter upper bounds for the probability in Problem \ref{safe}  compared to \eqref{ebr_2} with $\alpha=1$ (it also corresponds to the $c-$martingale in \cite{mathiesen2022safety}) and \eqref{ebr_33} with $\alpha=1.1$. Similarly, Table \ref{tab:my_label4} shows that the constraint \eqref{ebr_4} with $\alpha=\frac{1}{1.1}$  outperforms \eqref{ebr_4} with $\alpha=1$ and \eqref{ebr_41} with $\alpha=1.1$ in providing tighter bounds for the probability in Problem \ref{ravoid}. It is observed that employing higher-degree polynomials for computations facilitates the gain of tighter upper bounds of the probabilities in Problem \ref{safe} and \ref{ravoid}.

\begin{table*}[h!]
    \centering
    \begin{adjustbox}{width=\textwidth}
    \begin{tabular}{|c|c|c|c|c|c|c|c|c|c|c|c|}\hline
      \multicolumn{11}{|c|}{\eqref{ebr_33} with $\alpha=1.1$}\\\hline
       d &2&4&6&8&10&12&14&16&18&20  \\\hline
        $\epsilon_2$&0.9795& 0.9435& 0.9427& 0.9427& 0.9427& 0.9427& 0.9427& 0.9427& 0.9427& 0.9427 \\\hline  
        \multicolumn{11}{|c|}{\eqref{ebr_2} with $\alpha=\frac{1}{1.1}$} \\\hline
           d &2&4&6&8&10&12&14&16&18&20  \\\hline
        $\epsilon_2$& 0.8166& 0.1564& 0.0650& 0.0447& 0.0404& 0.0398& 0.0398&  0.0398& 0.0398& 0.0398\\\hline 
        
         \multicolumn{11}{|c|}{\eqref{ebr_2} with $\alpha=1$}\\\hline
           d &2&4&6&8&10&12&14&16&18&20  \\\hline
        $\epsilon_2$&0.1351& 0.1351& 0.1351& 0.1351& 0.1351& 0.1351& 0.1351& 0.1351& 0.1351&0.1351\\\hline
    \end{tabular}
    \end{adjustbox}
          \caption{\centering Computed upper bounds of the probability in Problem \ref{safe} in Example \ref{ex2} \\($d$ denotes the degree of the polynomial $v(\bm{x})$)}
     \label{tab:my_label3}
\end{table*}

\end{example}

\begin{table*}[htbp!]
    \centering
    \begin{adjustbox}{width=\textwidth}
    \begin{tabular}{|c|c|c|c|c|c|c|c|c|c|c|c|}\hline 
   \multicolumn{11}{|c|}{\eqref{ebr_41} with $\alpha=1.1$}\\\hline
    d &2&4&6&8&10&12&14&16&18&20  \\\hline
        $\epsilon_2$&0.9943& 0.9553&0.9428&0.9427&0.9427& 0.9427&0.9427&0.9427&0.9427& 0.9427 \\\hline  

        \multicolumn{11}{|c|}{\eqref{ebr_4} with $\alpha=\frac{1}{1.1}$}\\\hline
           d &2&4&6&8&10&12&14&16&18&20  \\\hline
              $\epsilon_2$&1& 0.2530& 0.1260&  0.0970& 0.0906& 0.0898& 0.0897&0.0897& 0.0897& 0.0897 \\\hline    
    \multicolumn{11}{|c|}{\eqref{ebr_4} with $\alpha=1$}\\\hline
           d &2&4&6&8&10&12&14&16&18&20  \\\hline
              $\epsilon_2$&0.1736&0.1736&0.1736&0.1736& 0.1736&0.1736&0.1736&0.1736&0.1736&0.1736 \\\hline  
    \end{tabular}
    \end{adjustbox}
    \caption{\centering Computed upper bounds of the probability in Problem \ref{ravoid} in Example \ref{ex2} \\($d$ denotes the degree of the polynomial $v(\bm{x})$)}
     \label{tab:my_label4}
\end{table*}

\begin{table*}[h!]
    \centering
    \begin{adjustbox}{width=\textwidth}
    \begin{tabular}{|c|c|c|c|c|c|c|c|c|c|c|c|}\hline
      \multicolumn{11}{|c|}{\eqref{ebr_33} with $\alpha=1.01$}\\\hline
       d &2&4&6&8&10&12&14&16&18&20  \\\hline
        $\epsilon_2$&0.8798& 0.7694& 0.7423& 0.7302& 0.6923& 0.6663& 0.6130& 0.6127&  0.5837& 0.5830 \\\hline  
        
        \multicolumn{11}{|c|}{\eqref{ebr_33} with $\alpha=1.001$}\\\hline
         d &2&4&6&8&10&12&14&16&18&20  \\\hline
        $\epsilon_2$&0.8169&  0.6612&  0.6041& 0.5845& 0.5316&  0.4942& 0.4131& 0.4098&  0.3615&  0.3605 \\\hline
        \multicolumn{11}{|c|}{\eqref{ebr_2} with $\alpha=\frac{1}{1.01}$} \\\hline
           d &2&4&6&8&10&12&14&16&18&20  \\\hline
        $\epsilon_2$&1.0000& 1.0000& 0.9625& 0.9235& 0.8352& 0.7731& 0.6342& 0.6236& 0.5424& 0.5397\\\hline 
        
        \multicolumn{11}{|c|}{\eqref{ebr_2} with $\alpha=\frac{1}{1.001}$}\\\hline
           d &2&4&6&8&10&12&14&16&18&20  \\\hline
        $\epsilon_2$&0.8515& 0.6877& 0.6179& 0.5929& 0.5381& 0.4983& 0.4097& 0.4027& 0.3505&0.3488\\\hline 
         \multicolumn{11}{|c|}{\eqref{ebr_2} with $\alpha=1$}\\\hline
           d &2&4&6&8&10&12&14&16&18&20  \\\hline
        $\epsilon_2$&0.8100& 0.6542& 0.5880& 0.5643& 0.5123& 0.4745& 0.3902& 0.3835& 0.3345&0.3322\\\hline 
  \multicolumn{11}{|c|}{\eqref{ebr_e7} with $\alpha=1.1$ and $\beta=0$ }\\\hline
   d &2&4&6&8&10&12&14&16&18&20  \\\hline
        $\epsilon_1$&$2.5096\times 10^{-14}$&0.0279& 0.0490& 0.0502& 0.0711& 0.0735&  0.0990& 0.1064& 0.1245& 0.1289 \\\hline 
    \end{tabular}
    \end{adjustbox}
      \caption{\centering Computed lower and upper bounds of the probability in Problem \ref{safe} in Example \ref{ex1} \\($d$ denotes the degree of the polynomial $v(\bm{x})$)}
     \label{tab:my_label}
\end{table*}

\begin{table*}[h!]
    \centering
     \begin{adjustbox}{width=\textwidth}
    \begin{tabular}{|c|c|c|c|c|c|c|c|c|c|c|c|}\hline 
   \multicolumn{11}{|c|}{\eqref{ebr_41} with $\alpha=1.001$}\\\hline
    d &2&4&6&8&10&12&14&16&18&20  \\\hline
        $\epsilon_2$&1& 0.9485& 0.9388&0.9264&0.9117&0.8988&0.8832&0.8682&0.8579&0.8513 \\\hline  

         \multicolumn{11}{|c|}{\eqref{ebr_41} with $\alpha=1.0001$}\\\hline
 d &2&4&6&8&10&12&14&16&18&20  \\\hline
        $\epsilon_2$&1& 0.9463&0.9360&0.9263&0.9077&0.8941&0.8780&0.8623&0.8515&0.8446 \\\hline 
                 \multicolumn{11}{|c|}{\eqref{ebr_41} with $\alpha=1$}\\\hline
 d &2&4&6&8&10&12&14&16&18&20  \\\hline
        $\epsilon_2$&1.0000& 0.9460& 0.9357& 0.9259& 0.9073&  0.8936&  0.8774& 0.8617& 0.8508& 0.8439 \\\hline  
        \multicolumn{11}{|c|}{\eqref{ebr_4} with $\alpha=\frac{1}{1.001}$ }\\\hline
           d &2&4&6&8&10&12&14&16&18&20  \\\hline
        $\epsilon_2$&1.0000& 0.9929& 0.9827& 0.9717& 0.9524&  0.9380&  0.9208& 0.9043& 0.8932& 0.8858\\\hline 

        \multicolumn{11}{|c|}{\eqref{ebr_4} with $\alpha=\frac{1}{1.0001}$}\\\hline
           d &2&4&6&8&10&12&14&16&18&20  \\\hline
        $\epsilon_2$&1.0000& 0.9506& 0.9403& 0.9304& 0.9117&  0.8979&  0.8816& 0.8658& 0.8550& 0.8480\\\hline 
        \multicolumn{11}{|c|}{\eqref{ebr_4} with $\alpha=1$}\\\hline
           d &2&4&6&8&10&12&14&16&18&20  \\\hline
        $\epsilon_2$&1.0000& 0.9460& 0.9357& 0.9259& 0.9073&  0.8936&  0.8774& 0.8617& 0.8508& 0.8439\\\hline 
   
  \multicolumn{11}{|c|}{\eqref{ebr_e6} with $\alpha=1.06$ and $\beta=0$}\\\hline
   d &2&4&6&8&10&12&14&16&18&20  \\\hline
        $\epsilon_1$& 0.1591& 0.2824& 0.3453&0.3669& 0.4606& 0.5218&  0.5732& 0.5778& 0.6119& 0.6128 \\\hline 
    \end{tabular}
    \end{adjustbox}
     \caption{\centering Computed lower and upper bounds of the probability in Problem \ref{ravoid} in Example \ref{ex1} \\($d$ denotes the degree of the polynomial $v(\bm{x})$)}
     \label{tab:my_label2}
\end{table*}

\begin{example}
\label{ex1}
Consider the following one-dimensional discrete-time system from \cite{yu2023safe}:
\begin{equation*}
x(l+1)=(-0.5+d(l))x(l),
\end{equation*}
where $d(\cdot)\colon\mathbb{N}\rightarrow \Theta=[-1,1]$, $\mathcal{X}=\{\,x\mid h(x)\leq 0\,\}$ with $h(x)=x^2-1$, $\bm{x}_0=-0.9$, $\mathcal{X}_r=\{\,x\mid x^2-0.36\leq 0\,\}$, and $N=50$. Besides, we assume that the probability distribution on $\Theta$ is the uniform distribution. The probabilities in Problem \ref{safe} and \ref{ravoid} obtained via Monte Carlo methods, which used $10^4$ sample paths, are around 0.2321 and 0.7708, respectively.

The set $\widetilde{\mathcal{X}}=\widehat{\mathcal{X}}=\{\,x\mid x^2-2.25\leq 0\,\}$ is used in  solving \eqref{ebr_2}, \eqref{ebr_33}, \eqref{ebr_e7}, \eqref{ebr_4}, \eqref{ebr_41}, and \eqref{ebr_e6}.  These constraints are addressed via encoding them into semi-definite programs with sum of squares decomposition techniques for multivariate polynomials. Utilizing polynomials $v(\bm{x})$ of varying degrees, the computed lower and upper bounds are summarized in Tables \ref{tab:my_label} and \ref{tab:my_label2}. It is evident from the results that employing higher-degree polynomials for computations leads to tighter lower and upper bounds of the probabilities in Problems \ref{safe} and \ref{ravoid}, and the constraints \eqref{ebr_2} and \eqref{ebr_33} can complement each other in providing upper bounds of the probability in Problem \ref{safe}. However, the performance of the constraint \eqref{ebr_41}, with $\alpha=\frac{1}{1.001}$, $\alpha=\frac{1}{1.0001}$, marginally surpasses that of \eqref{ebr_4}, with $\alpha=1.001$, $\alpha=1.0001$, in yielding tighter upper bounds for the probability in Problem \ref{ravoid}.

\end{example}

In  Example \ref{ex2} and \ref{ex1}, we initially determine upper bounds by assigning pre-defined values to the parameter $\alpha$ in constraints \eqref{ebr_2}, \eqref{ebr_33},  \eqref{ebr_4} and \eqref{ebr_41}, and lower bounds by assigning pre-defined values to the parameters $\alpha$ and $\beta$ in \eqref{ebr_e7} and \eqref{ebr_e6}, and employing convex optimization to solve them. This might yield conservative bounds. However, the automatic optimization of the function $v(\bm{x})$, $\alpha$, and $\beta$ to enhance these bounds constitutes a non-convex problem, which poses a significant challenge. Future work will address this issue, as it is beyond the scope of the current study. \textit{Furthermore, this work does not concentrate on the design of efficient algorithms to solve the associated constraints. Instead, it leverages existing semi-definite programming tools to tackle these issues. Future work will address this gap by proposing efficient algorithms to effectively address these constraints.}

\section{Conclusion}
\label{sec:con}
In this paper, we introduced novel sufficient conditions for the finite-time safety and reach-avoid verification of stochastic discrete-time dynamical systems. These conditions provide the lower and upper bounds of the probability that, within a predefined finite-time horizon, a system starting from an initial state in a safe set will either exit the safe set (safety verification) or reach a target set while remaining within the safe set until the first encounter with the target (reach-avoid verification). They complement existing criteria or bridge existing gaps in the literature. Finally, we demonstrated their performance in finite-time safety and reach-avoid verification on two numerical examples, utilizing semi-definite programming tools.

In the future, I would like to rigorously assess the conservativeness of the derived bounds through a necessity analysis of the proposed barrier-like conditions, extending the infinite-time framework developed in \cite{xue2024sufficient}. Furthermore, I would develop analysis methods for probabilistic programs by integrating these barrier-like conditions. For example, I intend to investigate termination analysis of probabilistic programs within bounded time horizons and compare the results with state-of-the-art approaches, such as those presented in \cite{chatterjee2024quantitative}.
\section*{Acknowledgements}
This work is funded by the CAS Pioneer Hundred Talents Program, Basic Research Program of  Institute of Software, CAS (Grant No. ISCAS-JCMS-202302), and NRF RSS Scheme NRF-RSS2022-009.
\bibliographystyle{abbrv}
\bibliography{reference}
\end{document}